\documentclass[twocolumn,10pt]{IEEEtran}
\usepackage{ifpdf, flushend,subfigure}

\ifCLASSINFOpdf
  \usepackage[pdftex]{graphicx}
  \graphicspath{{../pdf/}{../jpeg/}}
  \DeclareGraphicsExtensions{.pdf,.jpeg,.png}
\else
  \usepackage[dvips]{graphicx}
  \graphicspath{{../eps/}}
  \DeclareGraphicsExtensions{.eps}
\fi
\usepackage{graphicx}
\usepackage{epstopdf}
\graphicspath{{../}}
\usepackage{multirow}
\usepackage{float}
\usepackage[cmex10]{amsmath}
\usepackage{amsthm}
\usepackage{tabularx}

\theoremstyle{plain}
\newtheorem{lemma}{Lemma}

\theoremstyle{remark}

\usepackage {amssymb}
\usepackage{array}
\usepackage{mdwmath}
\usepackage{mdwtab}
\usepackage{eqparbox}
\usepackage{nomencl}
\usepackage{cite}
\usepackage{algorithm}
\usepackage{algorithmic}
\usepackage{makeidx}
\usepackage{ifthen}
\usepackage{subfigure}
\usepackage{caption}
\usepackage{pdflscape}
\usepackage{hyperref}
\usepackage{xcolor}
\usepackage{makecell}
\newcommand{\beq}{\begin{equation}}
\newcommand{\eeq}{\end{equation}}
\newcommand{\beqn}{\begin{eqnarray}}
\newcommand{\eeqn}{\end{eqnarray}}
\newcommand{\beqno}{\begin{eqnarray*}}
\newcommand{\eeqno}{\end{eqnarray*}}
\newcommand{\bma}{\begin{displaymath}}
\newcommand{\ema}{\end{displaymath}}
\newcommand{\bnu}{\begin{enumerate}}
\newcommand{\enu}{\end{enumerate}}
\newcommand{\bce}{\begin{center}}
\newcommand{\ece}{\end{center}}
\newcommand{\btb}{\begin{tabular}}
\newcommand{\etb}{\end{tabular}}

\usepackage{array}
\usepackage{makecell}

\usepackage{tikz}
\usetikzlibrary{arrows.meta, positioning, shapes.geometric}
\begin{document}






\title{DLMP-Based Bilevel Coordination of EV Charging and Reactive Power Support in Distribution Networks Under Uncertainty}

\title{Scenario-Free Uncertainty-Aware DLMP-Based Bilevel Coordination of EV Charging and Reactive Power Support in Distribution Networks}

\author{\IEEEauthorblockN{Arash~Baharvandi,~\IEEEmembership{Student Member,~IEEE} and~Duong~Tung~Nguyen,~\IEEEmembership{Member,~IEEE}}  
\thanks{The authors are with the School of Electrical, Computer and Energy Engineering, Arizona State University, Tempe, AZ, United States. Email: \textit\{abaharv1,~duongnt\}@asu.edu.
 }
 }


\maketitle

\begin{abstract}
This paper develops a scenario-free uncertainty-aware bilevel optimization framework for coordinated electric vehicle (EV) charging and reactive power support in distribution networks using distribution locational marginal prices (DLMPs). The upper-level EV aggregator jointly schedules active and reactive charging power to minimize charging costs, while the lower-level energy management system performs network-constrained economic dispatch and determines DLMPs subject to feeder and voltage constraints. To capture uncertainties in load demand and photovoltaic (PV) generation, a compact robust counterpart (RC) reformulation is developed that avoids the computational burden of large-scale stochastic programming and conventional robust optimization. Unlike existing robust counterpart methods that primarily assume Gaussian uncertainties, the proposed approach derives a deterministic reformulation for net-demand uncertainty modeled by a normal-minus-beta distribution, providing a more realistic representation of asymmetric load and renewable variability. An exactness lemma preserves the economic interpretation of DLMPs after KKT reformulation and Big-M linearization. EV chargers also provide reactive power support through non-unity power factor operation to improve voltage regulation. Simulation results on the IEEE 33-bus distribution system demonstrate improved voltage security, effective uncertainty-aware EV coordination, and significantly lower computational complexity than conventional stochastic and robust optimization approaches.
\end{abstract}

\begin{IEEEkeywords}
Electric vehicle charging, bilevel optimization, distribution locational marginal prices, uncertainty-aware optimization, reactive power support, distribution networks.
\end{IEEEkeywords}

\printnomenclature


\section{Introduction}

The rapid growth of electric vehicles (EVs) is significantly increasing electricity demand in distribution networks. Large-scale deployment of residential, commercial, and public charging infrastructure introduces substantial spatial and temporal variations in load, creating new operational challenges such as feeder congestion, voltage violations, increased power losses, and higher operating costs \cite{senapati2024advancing, pandey2024enhanced}. Thus, coordinated EV charging has become a key component of modern distribution system operation, enabling transportation electrification while maintaining network reliability and operational efficiency.

Advanced energy management systems (EMSs) increasingly coordinate distributed energy resources (DERs), renewable generation, and flexible demand through optimal power flow (OPF). In OPF-based operation, distribution locational marginal prices (DLMPs) provide economically meaningful price signals that reflect network congestion, voltage constraints, and operating conditions. By responding to these location-dependent prices, EV aggregators can schedule charging demand to reduce electricity procurement costs while supporting secure distribution system operation.

Recent studies have investigated DLMP- and LMP-based charging strategies, market-based coordination mechanisms, and reactive power support from EV chargers to improve voltage regulation and overall grid performance \cite{kazemtarghi2022optimal,ke2025response,song2025pricing,mazumder2020ev}.  These studies demonstrate the technical feasibility of EV-based grid services and price-responsive charging. However, most existing approaches rely on deterministic optimization or treat electricity prices as exogenous signals, limiting their ability to capture the interaction between EV aggregators and distribution system operators \cite{liu2020optimal,dolgui2024scheduling,elghanam2024optimization,han2017optimal}.

 Uncertainty in EV charging has been widely addressed using stochastic programming (SP), robust optimization (RO), and distributionally robust optimization (DRO). Stochastic approaches explicitly model uncertain EV arrivals, renewable generation, electricity prices, and charging demand \cite{chen2024optimal,fallah2020charge,luo2017stochastic }, whereas robust and distributionally robust methods improve operational reliability under uncertain demand and market conditions \cite{ren2024two,shi2022day,korolko2015robust}.  Although these methods improve scheduling performance under uncertainty, SP often suffers from scalability issues due to large scenario sets, while RO and DRO may introduce conservative decisions together with additional uncertainty-set constraints and reformulation complexity.


Recent studies have employed bilevel and hierarchical optimization frameworks to coordinate EV charging through endogenous electricity prices \cite{wang2023tri,gao2021multiagent,meng2024distributed,nasiri2023moment,patnam2020dlmp}. Although these methods effectively model the interaction between EV aggregators and distribution system operators, they generally rely on scenario-based SP or conventional RO, resulting in computationally demanding formulations. Moreover, existing robust counterpart methods primarily assume Gaussian uncertainty distributions and do not adequately represent the asymmetric uncertainty associated with photovoltaic (PV) generation.


To address these limitations, this paper develops a scenario-free uncertainty-aware bilevel optimization framework for coordinated EV charging and reactive power support in distribution networks using DLMPs. The upper-level EV aggregator jointly schedules active and reactive charging power, while the lower-level EMS performs network-constrained economic dispatch and determines endogenous DLMPs. To capture uncertainties in load demand and PV generation, a compact robust counterpart (RC) reformulation is developed that eliminates the need for large stochastic scenario sets while preserving computational tractability. Unlike existing robust counterpart methods, the proposed approach derives a deterministic reformulation for net-demand uncertainty modeled by a normal-minus-beta distribution, providing a more realistic representation of asymmetric renewable and load variability. Furthermore, an exactness lemma establishes that the relaxed uncertainty-aware power-balance constraints remain binding at optimality, preserving the economic interpretation of DLMPs after KKT reformulation and Big-M linearization. The framework also enables EV chargers to provide reactive power support through non-unity power factor operation, thereby improving voltage regulation while satisfying charging requirements. Our main contributions are summarized as follows:



\begin{itemize}

\item \textbf{\textit{Scenario-free uncertainty modeling:}} A compact robust counterpart is derived for net-demand uncertainty, modeled by a normal-minus-beta distribution. The proposed formulation captures asymmetric uncertainty without relying on large stochastic scenario sets or high-dimensional uncertainty budgets.

\item \textbf{Theoretical guarantee}: An exactness lemma proves that the relaxed uncertainty-aware power-balance constraints are binding at optimality, preserving the economic interpretation of DLMPs and establishing the exactness of the proposed deterministic reformulation.


\item \textbf{\textit{DLMP-based EV coordination:}} A bilevel optimization framework jointly schedules active and reactive EV charging based on endogenous DLMPs determined by the EMS.


\item \textbf{\textit{Grid-supportive EV operation:}} Non-unity-power-factor EV charging enables coordinated active and reactive power control, improving feeder voltage regulation.

\item \textbf{\textit{Numerical validation:}} Case studies on the IEEE 33-bus distribution system demonstrate improved voltage security and computational efficiency over conventional stochastic and robust optimization approaches.

\end{itemize}

The remainder of the paper is organized as follows: Section \ref{system model} outlines the system modeling framework, while Section \ref{sec:formu} formulates the bilevel optimization problem and the uncertainty modeling approach. Section \ref{Results} provides and discusses the simulation results. Finally, Section \ref{Conclusion} concludes the paper. 
\begin{figure}[h!]
	\centering
		\includegraphics[width=0.36\textwidth,height=0.17\textheight]{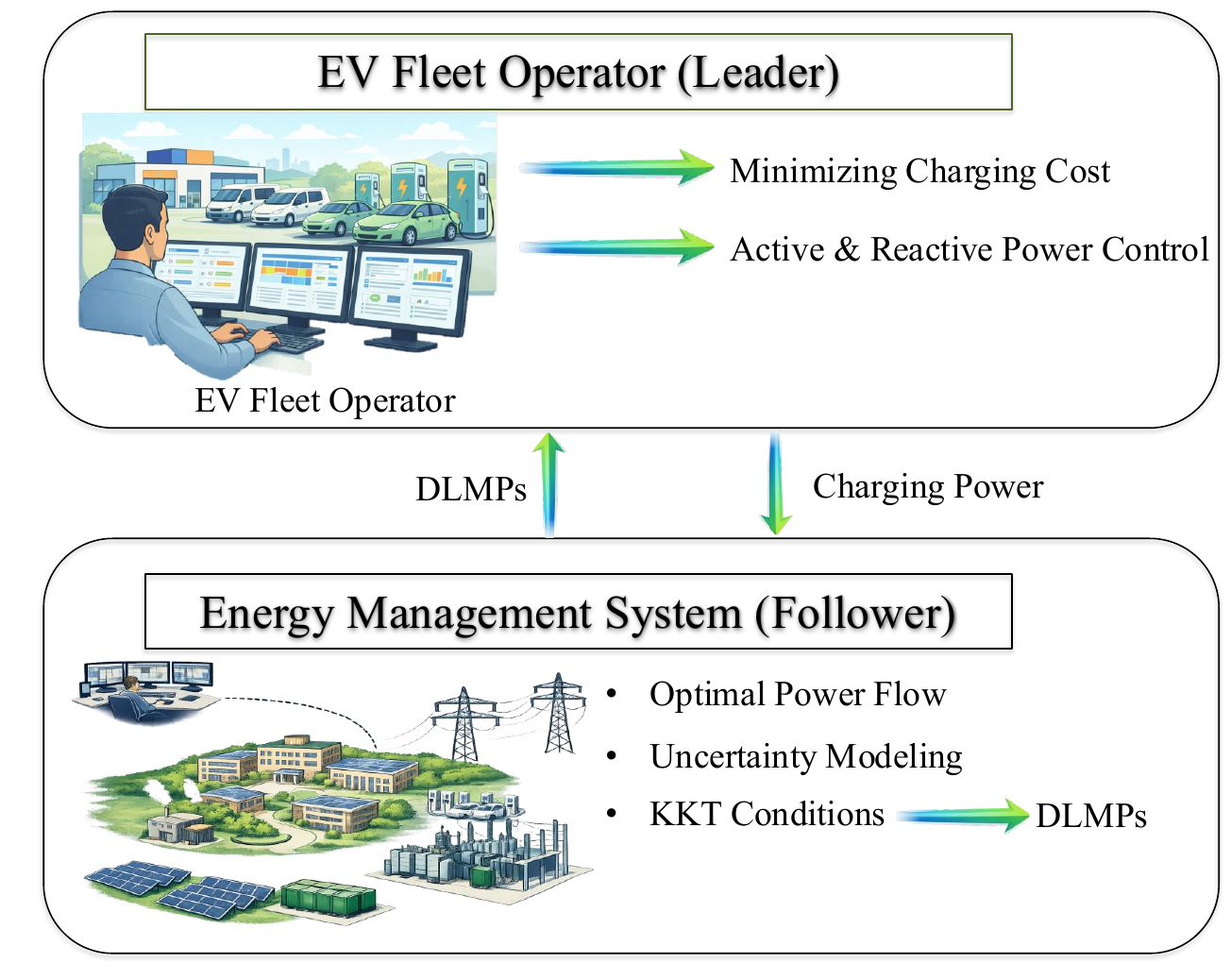}
				\caption{Illustration of the bilevel framework and the interaction between the leader and the follower}
		\label{fig:model}
\end{figure} 

\begin{table}[!h] 
\centering
\caption{Notations}
\begin{tabular}{|l|l|}
\hline
Notation   & Meaning\\
\hline	
\multicolumn{2}{|c|}{\textbf{Set and indices}}\\
\hline	
$\mathcal{I}^{b}$& Set of DGs or substations at bus $b$\\
\hline
$\mathcal{B}^{EV}$& Set of buses including charging stations 
\\
\hline	
$\mathcal{E}^{b}$, $e$ & Set of EVs at bus $b$ and index of EVs\\
\hline	
$\mathcal{T}$,$t$   & Set and index of periods \\
\hline 
$\mathcal{B},b$ & Set and index of buses\\
\hline	
$\mathcal{I}$,$i$ & Set and index of distributed generator or substation\\
\hline	
$\mathcal{L}$,$l$ & Set and index of lines\\
\hline
$s(l), r(l)$ & Sending bus and receiving bus of line $l$\\
\hline
\multicolumn{2}{|c|}{\textbf{Parameters}}\\
\hline
$\alpha_t$ & Duration of period $t$\\
\hline
$\varrho_{e,b}$ & EV charging efficiency coefficient\\
\hline
$E_{e,b}$ & Battery capacity of EV $e$ at bus $b$ \\
\hline
$S^E_{e,b}$&Socket rating of EV $e$ at bus $b$\\
\hline
$SOC^d_{e,b}$ & Desired SOC for EV $e$ at bus $b$ \\
\hline
$SOC^{min}_{e,b}$ & Minimum SOC for EV $e$ at bus $b$ \\
\hline
$SOC^{max}_{e,b}$ & Maximum SOC for EV $e$ at bus $b$ \\
\hline
$d^{p}_{b,t}$ & Forecast active demand at bus $b$ and period $t$ \\
\hline
$d^{q}_{b,t}$ & Forecast reactive demand at bus $b$ and period $t$ \\
\hline
$P^{pv}_{b,t}$ & Predicted output power of PV at bus $b$ and period $t$ \\
\hline
$R_l, ~X_l$& Resistance/Reactance of line $l$ \\
\hline
$P^{max}_{i}$ & Maximum active capacity of DG or substation $i$\\
\hline
$P^{min}_{i}$ & Minimum active capacity of DG or substation $i$\\
\hline
$Q^{max}_{i}$ & Maximum reactive capacity of DG or substation $i$\\
\hline
$Q^{min}_{i}$ & Minimum reactive capacity of DG or substation $i$\\
\hline
$P^{max}_{l}$ & Maximum active capacity of line $l$\\
\hline
$Q^{max}_{l}$ & Maximum reactive capacity of line $l$\\
\hline
$V^{min}_{b}$ & Minimum squared voltage at bus $b$ \\
\hline
$V^{max}_{b}$ & Maximum squared voltage at bus $b$ \\
\hline
$\gamma$, $\delta$, 1-$\Theta$ &Infeasibility tolerance, uncertainty and reliability level\\
\hline
\multicolumn{2}{|c|}{\textbf{Variables}}\\

\hline	
$\mu_{b,t}$ &DLMP at bus $b$ and period $t$\\
\hline	
$C_{i,t}$ &Quadratic generation cost function\\
\hline
$P^{E}_{e,b,t} $ & Charging power for EV $e$ at bus $b$ and period $t$\\
\hline
$Q^{E}_{e,b,t} $ & Injected reactive power by EV $e$ at bus $b$ and period $t$\\
\hline
$SOC_{e,b,t}$ &State of charge for EV $e$ at bus $b$ and period $t$\\
\hline
$P^f_{b,t}$ &Total active power of EVs at bus $b$ and period $t$  \\
\hline
$Q^f_{b,t}$ &Total reactive power of EVs at bus $b$ and period $t$  \\
\hline
$P^G_{i,t}$ &Output active power of DG $i$ at period $t$\\
\hline
$Q^G_{i,t}$ &Output reactive power of DG $i$ at period $t$\\
\hline
$P^L_{l,t}$ &Active power flow in line $l$ at period $t$\\
\hline
$Q^L_{l,t}$ &Reactive power flow in line $l$ at period $t$\\
\hline
$V_{b,t}$ &Squared Voltage at bus $b$ and period $t$\\
\hline

\end{tabular} 
\label{notation}
\end{table}

\section{System Model}
\label{system model}
This section presents the hierarchical architecture of the proposed uncertainty-aware bilevel framework for coordinated EV charging in distribution networks, illustrated in Fig.~\ref{fig:model}. The framework consists of two interacting layers: an \textit{EV aggregator} (leader) and a centralized EMS (follower).
At the upper level, the EV aggregator optimizes the active and reactive charging power of EVs to minimize electricity procurement costs based on DLMPs determined by the EMS.

The optimization considers charger ratings, 
battery state-of-charge (SOC) dynamics, departure SOC requirements, and inverter operating constraints. EV chargers operate at non-unity power factor, enabling simultaneous active-power charging and reactive-power injection. Since reactive power capability is limited by the inverter apparent-power rating, active and reactive power are jointly optimized to satisfy charging requirements while supporting feeder voltage regulation.

Grid-connected power electronic converters enable EV chargers to provide reactive power support through fourth-quadrant operation. Modern Level-2 and DC fast chargers can independently regulate active and reactive power during grid-to-vehicle (G2V) charging within inverter capability limits. This capability is consistent with emerging smart charging standards and grid-interactive inverter requirements, including IEEE 1547 and ISO 15118.
At the lower level, the EMS solves a network-constrained economic dispatch problem to minimize generation cost subject to network constraints. It determines generator dispatch, branch power flows, bus voltages, and the associated DLMPs, which provide economic coordination signals for the EV aggregator.

To capture uncertainties in load demand and photovoltaic (PV) generation, the proposed framework adopts a \textit{scenario-free uncertainty model} based on a \textit{normal-minus-beta distributed} net-demand representation. Unlike conventional stochastic optimization, uncertainty is incorporated through a compact deterministic reformulation that avoids large scenario sets while preserving computational tractability. As discussed in Section III, the resulting formulation preserves the economic interpretation of DLMPs through the exactness of the relaxed power-balance constraints.

Let $\mathcal{I}^{b}$ denote the set of distributed generators (DGs) or substations at bus $b$, and $\mathcal{E}^{b}$ the set of EVs connected to bus $b$, with indices $i$ and $e$, respectively. The sets of time periods, buses, DGs/substations, and distribution lines are denoted by $\mathcal{T}$, $\mathcal{B}$, $\mathcal{I}$, and $\mathcal{L}$, with corresponding indices $t$, $b$, $i$, and $l$. The set of buses equipped with EV charging stations is denoted by $\mathcal{B}^{EV}$. For each line $l$, $s(l)$, and $r(l)$ represent the sending and receiving buses, respectively, and $\alpha_t$ denotes the duration of period $t$.
For each EV at bus $b$, $\varrho_{e,b}$, $E_{e,b}$, and $S^{E}_{e,b}$ denote the charging efficiency, battery capacity, and charger rating, respectively. The desired, minimum, and maximum states of charge (SOC) are denoted by $SOC^{d}_{e,b}$, $SOC^{\min}_{e,b}$, and $SOC^{\max}_{e,b}$, respectively.
The forecast active and reactive load demands at bus $b$ and time $t$ are denoted by $d^{p}_{b,t}$ and $d^{q}_{b,t}$, respectively, while the forecast PV active power output is $P^{pv}_{b,t}$. The resistance and reactance of line $l$ are $R_l$ and $X_l$, respectively. For each DG or substation unit $i$, the active and reactive power limits are bounded by $P^{\min}_{i}$ and $P^{\max}_{i}$, $Q^{\min}_{i}$, and $Q^{\max}_{i}$, respectively. The active and reactive power flow limits of line $l$ are $P^{\max}_{l}$ and $Q^{\max}_{l}$, respectively. Bus voltage magnitudes are constrained by $V^{\min}_{b} \le V_{b,t} \le V^{\max}_{b}$.

Parameters $\gamma$, $\delta$, and $\Theta$ denote the infeasibility tolerance, uncertainty level, and reliability level, respectively. The DLMP at bus $b$ and time $t$ is denoted by $\mu_{b,t}$, and the quadratic generation cost of DG $i$ at time $t$ is represented by $C_{i,t}(\cdot)$. The active and reactive charging powers of EV $e$ at bus $b$ and time $t$ are $P^{E}_{e,b,t}$ and $Q^{E}_{e,b,t}$, respectively, while the corresponding aggregated charging powers at bus $b$ are $P^{f}_{b,t}$ and $Q^{f}_{b,t}$. The active and reactive power outputs of DG $i$ are denoted by $P^{G}_{i,t}$ and $Q^{G}_{i,t}$, respectively, and the active and reactive line power flows are represented by $P^{L}_{l,t}$ and $Q^{L}_{l,t}$. Table~\ref{notation} summarizes the notations used throughout the paper.

\section{Problem Formulation}
\label{sec:formu}
\subsection{Deterministic Formulation}
\label{deterministic}
This section presents the deterministic bilevel optimization model for coordinated EV charging. The upper-level EV aggregator optimizes active and reactive charging schedules to minimize charging cost, while the lower-level EMS solves a network-constrained economic dispatch problem to determine the corresponding DLMPs.

\subsubsection{Upper-level}
\label{Upper level}
The EV aggregator minimizes charging cost based on the DLMPs determined by the lower-level EMS. The optimization jointly determines the active and reactive charging powers of individual EVs subject to charger ratings, battery dynamics, and SOC requirements. The upper-level problem is formulated as:
\begin{subequations}
\label{DET-upper}
\begin{align}
& \min_{{\small \boldsymbol{P}^E,\boldsymbol{Q}^E,\boldsymbol{SOC}}} \sum_{b \in \mathcal{B}^{EV}} \sum_{t\in \mathcal{T}}  \sum_{e \in \mathcal{E}^{b}} \mu_{b,t} P^{E}_{e,b,t} \alpha_t 
\label{det-obj}\\
& \textit{s.t:} ~~ 
SOC_{e,b,t} = SOC_{e,b,t-1}+ \frac{\varrho_{e,b} P^{ E}_{e,b,t} \alpha_t}{E_{e,b}},~\forall e,b,t 
\label{ev:soc}\\
& (P^{E}_{e,b,t})^{2}+(Q^{E}_{e,b,t})^{2}\leq 
(S^{E}_{e,b})^2,~~~\forall e,b,t 
\label{ev:pqs}\\
&SOC_{e,b,t}\geq SOC^d_{e,b},~~~\forall e,b,t = t_{e,b}^{d} 
\label{ev:desired}\\
&SOC^{min}_{e,b}\leq SOC_{e,b,t} \leq SOC^{max}_{e,b},~~~\forall e,b,t
\label{ev:soclimits}\\
&P^{f}_{b,t}=\sum_e P^{E}_{e,b,t} ,~~~\forall b,t 
\label{ev:PF}\\
&Q^{f}_{b,t}=\sum_e Q^{E}_{e,b,t} ,~~~\forall b,t 
\label{ev:QF}\\
&P^{E}_{e,b,t}\geq 0, Q^{E}_{e,b,t}\leq 0,~~~\forall e,b,t. 
\label{ev:PQF}
\end{align}
\end{subequations}

The objective function in (\ref{det-obj}) minimizes the total charging cost over the scheduling horizon, subject to constraints (\ref{ev:soc})–(\ref{ev:PQF}). 
The formulation considers individual EV travel patterns, including arrival and departure times, initial SOC, and desired SOC. Constraint (\ref{ev:soc}) models battery dynamics by updating the SOC according to the active charging power, while (\ref{ev:pqs}) limits the apparent power of each charger. Constraint (\ref{ev:desired}) ensures that each EV reaches its desired SOC by the scheduled departure time $t^{d}_{e,b}$. Battery health is preserved by restricting SOC to remain within allowable bounds, as defined in (\ref{ev:soclimits}). The aggregated active and reactive charging powers at bus $b$ are computed by (\ref{ev:PF}) and (\ref{ev:QF}), respectively. Finally, (\ref{ev:PQF}) defines the operating region of EV inverters, allowing EVs to draw active power while injecting reactive power to support voltage regulation.

\subsubsection{Lower-level}
\label{Lower level}

The EMS solves a network-constrained economic dispatch problem that minimizes total generation cost subject to operational and network constraints, given the flexible load specified by the EV aggregator. The lower-level problem is formulated as:
\begin{subequations}
\label{DET-lower}
\begin{align}
&\min_{\small\boldsymbol{P}^G,\boldsymbol{Q}^G,\boldsymbol{P}^L,\boldsymbol{Q}^L,\boldsymbol{V}}    \sum_{i \in \mathcal{I}} \sum_{t\in \mathcal{T}}   C_{i,t}(P^G_{i,t})
\label{det-objl}\\
& \textit{s.t:} ~~  
 \sum_{i\in \mathcal{I}^{b}} P^{G}_{i,t}-\sum_{l:s(l)=b} P^L_{l,t}+\sum_{l:r(l)=b} P^L_{l,t} \nonumber\\ 
 &= P^{f}_{b,t}+d^p_{b,t}-P^{pv}_{b,t},~~\forall b \in \mathcal{B},t
 \label{eq:dsopb} ~~~(\mu_{b,t})\\
&\sum_{i\in \mathcal{I}^{b}} Q^{G}_{i,t}-\sum_{l:s(l)=b} Q^L_{l,t}+\sum_{l:r(l)=b} Q^L_{l,t} \nonumber\\ 
&=Q^{f}_{b,t}+d^q_{b,t},~~\forall b,t
\label{eq:dsoqb} ~~~(\rho_{b,t})\\
&V_{s(l),t}-V_{r(l),t}\!=\!\!2(R_l P^L_{l,t}+X_l Q^L_{l,t}),~\forall l \in \mathcal{L},t
\label{eq:dsov} ~~(\omega_{l,t})\\
&P^{min}_i \leq P^{G}_{i,t} \leq P^{max}_i,~~\forall i,t
\label{eq:dsopgl} ~~~(\lambda^{p-}_{i,t},\lambda^{p+}_{i,t})\\
&Q^{min}_i \leq Q^{G}_{i,t} \leq Q^{max}_i,~~\forall i,t
\label{eq:dsoqgl} ~~~(\lambda^{q-}_{i,t},\lambda^{q+}_{i,t})\\
&-P^{max}_{l} \leq P^{L}_{l,t} \leq P^{max}_{l},~~\forall l,t
\label{eq:dsopll}~~~(\tau^{p-}_{l,t},\tau^{p+}_{l,t})\\
&-Q^{max}_{l} \leq Q^{L}_{l,t} \leq Q^{max}_{l},~~\forall l,t
\label{eq:dsoqll}~~~(\tau^{q-}_{l,t},\tau^{q+}_{l,t})\\
&V^{min}_b\leq V_{b,t}\leq V^{max}_b,~~~\forall b,t. 
\label{eq:dsovl}~~~(\sigma^{-}_{b,t},\sigma^{+}_{b,t})
\end{align}
\end{subequations}
The EMS minimizes the generation cost in (\ref{det-objl}), with the upstream grid modeled as a high-cost source to prioritize local DGs and highlight the effects of local generation and network constraints on DLMP formation. This representation can be extended to include time-varying wholesale prices or contractual tariffs. Constraints (\ref{eq:dsopb})–(\ref{eq:dsoqb}) enforce nodal active and reactive power balance, incorporating the aggregated EV powers $P^{f}_{b,t}$ and $Q^{f}_{b,t}$ from the upper-level problem. The linearized DistFlow model, developed by Baran and Wu and further analyzed by Farivar and Low, is adopted for its efficiency and accuracy in moderately loaded radial networks \cite{Baran1989Capacitor, Farivar2013Part1}. Constraint (\ref{eq:dsov}) relates voltage drops to active and reactive flows while neglecting line losses \cite{Baran1989Capacitor,Farivar2013Part1,li5238685distributed,park2025fair}, so marginal loss components are not explicitly included in the DLMPs. Since uncertainty modeling, endogenous DLMP formation, KKT reformulation, and Big-$M$ linearization already produce a large-scale MPEC/MINLP, a full nonlinear AC OPF would significantly reduce tractability; thus, the linearized model provides a practical balance between accuracy and scalability, while extension to AC OPF is left for future work. Generator, line-flow, and voltage limits are enforced by (\ref{eq:dsopgl})–(\ref{eq:dsoqgl}), (\ref{eq:dsopll})–(\ref{eq:dsoqll}), and (\ref{eq:dsovl}), respectively. The dual variables $\mu_{b,t}$, $\rho_{b,t}$, $\omega_{l,t}$, $\lambda^{p-}_{i,t}$, $\lambda^{p+}_{i,t}$, $\lambda^{q-}_{i,t}$, $\lambda^{q+}_{i,t}$, $\tau^{p-}_{l,t}$, $\tau^{p+}_{l,t}$, $\tau^{q-}_{l,t}$, $\tau^{q+}_{l,t}$, $\sigma^{-}_{b,t}$, and $\sigma^{+}_{b,t}$ correspond to constraints (\ref{eq:dsopb})–(\ref{eq:dsovl}), respectively. In particular, the dual variable $\mu_{b,t}$ associated with the active power balance constraint (\ref{eq:dsopb}) represents the DLMP at node $b$ and time $t$.

\vspace{-0.07cm}
\subsubsection{Single-level formulation}
\label{single level}
The bilevel optimization problem is reformulated as a single-level model by incorporating the Karush--Kuhn--Tucker (KKT) conditions of the lower-level problem into the upper-level problem. Since the lower-level problem is convex, the KKT conditions are both necessary and sufficient for optimality. The resulting formulation embeds the stationarity, primal feasibility, dual feasibility, and complementary slackness conditions, yielding an MPEC that can be solved using standard optimization solvers. The KKT conditions of the lower-level problem are given by: 
\vspace{-0.16cm}
\begin{subequations}
\label{KKT}
\begin{align}
&\frac{\partial C_{i,t}}{\partial P^G_{i,t}}-\mu_{b,t}-\lambda^{p-}_{i,t}+\lambda^{p+}_{i,t}=0,~~~\forall t, b, i \in \mathcal{I}^b
\label{dc1}\\
&\mu_{s(l),t}-\mu_{r(l),t}+2 R_l \omega_{l,t}-\tau^{p-}_{l,t}+\tau^{p+}_{l,t}=0,~~~\forall l,t
\label{dc2}\\
&\rho_{s(l),t}-\rho_{r(l),t}+2 X_l \omega_{l,t}-\tau^{q-}_{l,t}+\tau^{q+}_{l,t}=0,~~~\forall l,t
\label{dc3}\\
&-\rho_{b,t}-\lambda^{q-}_{i,t}+\lambda^{q+}_{i,t}=0,~~~\forall t, b, i \in \mathcal{I}^b
\label{dc4}\\
&\sum_{l:r(l)=b} \omega_{l,t}-\sum_{l:s(l)=b} \omega_{l,t}-\sigma^{-}_{b,t}+\sigma^{+}_{b,t}=0,~~~\forall b,t
\label{dc5}\\
&\lambda^{p-}_{i,t}(P^{min}_i-P^G_{i,t})=0,~~~\forall i,t
\label{dc6}\\
&\lambda^{p+}_{i,t}(P^G_{i,t}-P^{max}_{i})=0,~~~\forall i,t
\label{dc7}\\
&\lambda^{q-}_{i,t}(Q^{min}_i-Q^G_{i,t})=0,~~~\forall i,t
\label{dc8}\\
&\lambda^{q+}_{i,t}(Q^G_{i,t}-Q^{max}_{i})=0,~~~\forall i,t
\label{dc9}\\
&\tau^{p-}_{l,t}(-P^{max}_{l}-P^{L}_{l,t})=0,~~~\forall l,t
\label{dc10}\\
&\tau^{p+}_{l,t}(P^{L}_{l,t}-P^{max}_{l})=0,~~~\forall l,t
\label{dc11}\\
&\tau^{q-}_{l,t}(-Q^{max}_{l}-Q^{L}_{l,t})=0,~~~\forall l,t
\label{dc12}\\
&\tau^{q+}_{l,t}(Q^{L}_{l,t}-Q^{max}_{l})=0,~~~\forall l,t
\label{dc13}\\
&\sigma^{-}_{b,t}(V^{min}_{b}-V_{b,t})=0,~~~\forall b,t
\label{dc14}\\
&\sigma^{+}_{b,t}(V_{b,t}-V^{max}_{b})=0,~~~\forall b,t
\label{dc15}\\
&\boldsymbol{\lambda^{p-}}\!,\! \boldsymbol{\lambda^{p+}}\!, \!\boldsymbol{\lambda^{q-}}\!,\! \boldsymbol{\lambda^{q+}}\!, \!\boldsymbol{\tau^{p-}}\!, \!\boldsymbol{\tau^{p+}}\!,\! \boldsymbol{\tau^{q-}}\!, \!\boldsymbol{\tau^{q+}}\!, \!\boldsymbol{\sigma^{-}}\!,\! \boldsymbol{\sigma^{+}} \!\! \geq 0.
\label{dc16}
\end{align}
\end{subequations}
The stationarity conditions are given by (\ref{dc1})–(\ref{dc5}), the complementary slackness conditions by (\ref{dc6})–(\ref{dc15}), and the dual feasibility conditions by (\ref{dc16}). Consequently, the deterministic bilevel problem can be reformulated as the following single-level optimization model: \begin{equation}
\label{single-level-reformulation}
\Delta^d:= \{(\ref{DET-upper}), (\ref{eq:dsopb})-(\ref{eq:dsovl}), (\ref{KKT}) \}
\end{equation}
\vspace{-1.0cm}
\subsection{Uncertainty Approach}
\label{methodology}
This section presents the uncertainty modeling framework underlying the proposed approach. The formulation builds upon the robust counterpart method in \cite{lin2004new}, \cite{janak2007new} and is subsequently extended to accommodate normal-minus-beta distributed net-demand uncertainty. Consider the following generic deterministic optimization problem:
\begin{subequations}   
\label{DET-original}
\begin{align}
&\min_{\boldsymbol{x},\boldsymbol{y}} \boldsymbol{u}^{T} \boldsymbol{x}+\boldsymbol{f}^{T} \boldsymbol{y}
\label{obj-original} \\
& \textit{s.t:} ~~\boldsymbol{H}\boldsymbol{x}+\boldsymbol{M}\boldsymbol{y}\leq \boldsymbol{d} \label{originalc}
\\
&\boldsymbol{x}_{min} \leq \boldsymbol{x} \leq \boldsymbol{x}_{max}
\label{range}\\
&y_i \in \{0,1\},~~~\forall i. \label{binary range}
\end{align}
\end{subequations}

\noindent
In this formulation, $\boldsymbol{x}$ and $\boldsymbol{y}$ denote the decision variable vectors. The model parameters consist of vectors $\boldsymbol{u}$, $\boldsymbol{f}$, $\boldsymbol{d}$, $\boldsymbol{x}_{min}$, and $\boldsymbol{x}_{max}$, as well as matrices $\boldsymbol{H}$ and $\boldsymbol{M}$ of appropriate dimensions. Let index $i$ denote the elements of vector $\boldsymbol{y}$. In practice, the parameters $\boldsymbol{d}$, $\boldsymbol{H}$, and $\boldsymbol{M}$ are often subject to uncertainty due to variations in operating conditions and forecast errors. Assuming these parameters follow normal distributions, constraint (\ref{originalc}) can be reformulated as the following deterministic equivalent \cite{lin2004new,janak2007new}:
\begin{align}
&\sum_{j} h_{k,j} x_j+ \delta \varphi \sqrt{\sum_{j} h_{k,j}^{2} x_j^{2}+\sum_{i} m_{k,i}^{2} y_i+d_k^{2}}\label{modified} \nonumber\\ & +\sum_{i} m_{k,i} y_i\leq d_k+\gamma \max{\Big[1,|d_k|\Big]},~~~\forall k. 
\end{align}

Here, indices $k$ and $j$ refer to the elements of matrix $\boldsymbol{H}$, while $k$ and $i$ refer to those of matrix $\boldsymbol{M}$. The index $k$ also denotes the elements of vector $\boldsymbol{d}$. Parameters $h_{k,j}$, $m_{k,i}$, and $d_k$ represent the forecast values of the uncertain quantities, whose actual values are given by:
\begin{subequations}
\label{true values}
\begin{align}
&h_{k,j}^{act}=(1+\delta \chi_{k,j}) h_{k,j}
\label{true valueh}\\
&m_{k,i}^{act}=(1+\delta \chi_{k,i}) m_{k,i}
\label{true valuek}\\
&d_{k}^{act}=(1+\delta \chi_{k}) d_{k}. \label{true valuer}
\end{align}
\end{subequations}

Here, $h_{j,k}^{\mathrm{act}}$, $m_{k,i}^{\mathrm{act}}$, and $d_{k}^{\mathrm{act}}$ denote the actual realizations of the uncertain parameters. The parameters $\delta$ and $\varphi$ quantify the deviation of the actual values from their forecasts. Since the uncertain parameters are assumed to follow normal distributions, $\chi$ is modeled as a normally distributed random variable. In (\ref{modified}), $\delta$ and $\gamma$ denote the uncertainty level and the infeasibility tolerance, respectively. Larger values of $\delta$ imply greater robustness against uncertainty at the expense of increased conservatism. In contrast, larger values of $\gamma$ relax the feasibility requirements by allowing limited constraint violations, expanding the feasible region and generally reducing the objective cost. 
The deterministic reformulation in (\ref{modified}) is valid provided that
\cite{lin2004new,janak2007new}:
(i) the deterministic problem based on the forecasted parameters is feasible; and
(ii) the probability of constraint violation does not exceed a prescribed threshold, i.e.,
\begin{align}
&\Pr  \{\sum_{j} h_{k,j}^{act} x_j+\sum_{i} m_{k,i}^{act} y_i > d_k^{act} +\gamma \max{[1,|d_k|]} \}  \leq \Theta \nonumber
\end{align}
where $\varphi=F_n^{-1}(1-\Theta)$ and $1-\Theta$ determines reliability level. 
Furthermore, the relationship between $\varphi$ and $\Theta$ is given by:
\begin{align}
\label{reliability}
\Theta=1-\int\limits_{-\infty}^{\varphi} \frac{1}{\sqrt {2\pi}} e^{\sf \frac {-x^{\sf 2}} {2}}\,dx. 
\end{align}
In detailed proof is provided in \cite{janak2007new}. 


Unlike the uncertainty treatment in \cite{lin2004new,janak2007new}, which assumes that the uncertain parameter follows a single known probability distribution, the proposed framework models \textit{net-demand} uncertainty, defined as the difference between load demand and PV generation. Specifically, load demand is modeled by a normal distribution, whereas PV generation follows a beta distribution to capture its bounded and asymmetric characteristics. Consequently, the net demand is represented as the difference between two random variables with different probability distributions. Therefore, the deterministic reformulation cannot be obtained directly from the standard normal-distribution approach and requires a dedicated derivation for the resulting normal-minus-beta random variable.
To address this challenge, an equivalent deterministic reformulation of the chance constraint is derived. In (\ref{DET-original}), uncertainty is introduced only through the parameter vector $\boldsymbol{d}$, while $\boldsymbol{H}$ and $\boldsymbol{M}$ remain deterministic. This is consistent with the power-balance formulation, where uncertainty appears only on the right-hand side. Consider the following chance constraint:
\begin{align}
\Pr \{
\sum_{j} h_{k,j}x_j+\sum_{i} m_{k,i}y_i
\le \tilde d_k
\}\ge 1-\Theta 
\nonumber
\end{align}
where \(1-\Theta\) is the required reliability level and \(\Theta\) is the
allowable violation probability.
In the proposed formulation, the uncertain right-hand side is modeled as
the actual net demand $\tilde d_k =  D_k - PV_k $,
where the load demand \(D_k\) follows a normal distribution and the PV
generation \(PV_k\) follows a scaled beta distribution. Specifically,
$D_k \sim \mathcal N(\mu_{D,k},\sigma_{D,k}^2),
~
PV_k=s_k U_k,
~
U_k\sim \mathrm{Beta}(\alpha,\beta)$,
where \(D_k\) and \(U_k\) are independent. The scaling factor \(s_k\) maps
the beta-distributed variable from the normalized interval \([0,1]\) to the
physical PV-generation range. If $\mu_{PV,k}$ denotes the forecast or
expected PV generation, then
$s_k=\frac{\mu_{PV,k}}{\alpha/(\alpha+\beta)}$.
Substituting \(\tilde d_k=W_k\) into the chance constraint gives:
\begin{align}
\Pr \{
\sum_{j} h_{k,j}x_j+\sum_{i} m_{k,i}y_i
\le W_k
\}\ge 1-\Theta .
\nonumber
\end{align}
Define:
\begin{align}
z_k :=
\sum_{j} h_{k,j}x_j+\sum_{i} m_{k,i}y_i .
\label{z-definition}
\end{align}
Then the chance constraint is 
$
\Pr\{W_k \ge z_k\}\ge 1-\Theta .
$
Since \(W_k\) is continuous,
$
\Pr\{W_k \ge z_k\}
=
1-F_{W_k}(z_k),
$
where \(F_{W_k}(\cdot)\) is the cumulative distribution function (CDF) of
\(W_k\). Therefore,
$
1-F_{W_k}(z_k)\ge 1-\Theta
~~ \Longleftrightarrow ~~
F_{W_k}(z_k)\le \Theta $.
Since \(F_{W_k}(\cdot)\) is nondecreasing, the above condition is
equivalent to
$
z_k \le F_{W_k}^{-1}(\Theta).
$
Defining
$
q_{\Theta,k}:=F_{W_k}^{-1}(\Theta),
$
the deterministic equivalent of the chance constraint is obtained as:
\begin{align}
\sum_{j} h_{k,j}x_j+\sum_{i} m_{k,i}y_i
\le q_{\Theta,k}.
\label{equivalent}
\end{align}
Thus, the uncertain right-hand side is replaced by the deterministic
equivalent parameter \(q_{\Theta,k}\), which represents the
\(\Theta\)-quantile of the actual net demand \(W_k=D_k-PV_k\).
To compute \(q_{\Theta,k}\), the CDF of \(W_k\) is derived as follows:
\begin{align}
F_{W_k}(t)
=
\Pr(W_k\le t)
=
\Pr(D_k-PV_k\le t).
\nonumber
\end{align}
Since \(PV_k=s_kU_k\), we have
$
F_{W_k}(t)
=
\Pr(D_k-s_kU_k\le t).
$
Rearranging the inequality gives:
$
F_{W_k}(t)
=
\Pr(D_k\le t+s_kU_k).
$
Using the law of total probability with respect to \(U_k\), we obtain:
\begin{align}
F_{W_k}(t)
=
\int_{0}^{1}
\Pr(D_k\le t+s_ku \mid U_k=u) f_{U_k}(u)\,du .
\nonumber
\end{align}
Since \(D_k\) and \(U_k\) are independent:
\begin{align}
\Pr(D_k\le t+s_ku \mid U_k=u)
=
\Pr(D_k\le t+s_ku).
\nonumber
\end{align}
Because \(D_k\sim\mathcal N(\mu_{D,k},\sigma_{D,k}^2)\):
\begin{align}
\Pr(D_k\le t+s_ku)
=
\Phi\left(
\frac{t+s_ku-\mu_{D,k}}{\sigma_{D,k}}
\right),
\nonumber
\end{align}
where \(\Phi(\cdot)\) is the standard normal cumulative distribution
function.
Therefore:
\begin{align}
F_{W_k}(t)
=
\int_{0}^{1}
\Phi\left(
\frac{t+s_ku-\mu_{D,k}}{\sigma_{D,k}}
\right)
f_{U_k}(u)\,du .
\nonumber
\end{align}
The PDF of \(U_k\sim\mathrm{Beta}(\alpha,\beta)\) is:
\begin{align}
f_{U_k}(u)
=
\frac{
u^{\alpha-1}(1-u)^{\beta-1}
}{
B(\alpha,\beta)
},
\qquad 0\le u\le 1 .
\nonumber
\end{align}
Substituting the beta PDF yields:
\begin{align}
F_{W_k}(t)
=
\!\!\int_{0}^{1}
\!\!\!\Phi\!\left(
\frac{t+s_ku-\mu_{D,k}}{\sigma_{D,k}}
\right)
\!\frac{
u^{\alpha-1}(1-u)^{\beta-1}
}{
B(\alpha,\beta)
}
\,du .
\label{convolution}
\end{align}
The beta function is defined as
$
B(\alpha,\beta)
=
\frac{\Gamma(\alpha)\Gamma(\beta)}
{\Gamma(\alpha+\beta)},
$
where \(\Gamma(\cdot)\) denotes the gamma function.
Since the normal-minus-beta net-demand distribution does not generally admit a
closed-form inverse CDF, the deterministic equivalent parameter
\(q_{\Theta,k}\) is computed numerically from:
\begin{align}
F_{W_k}(q_{\Theta,k})=\Theta .
\label{q_k-calculation}
\end{align}
For example, the inverse CDF is computed offline using numerical quadrature together with a bisection search.
For each PV-equipped bus and time period with nonzero PV generation, $\mathbf{q}$ is precomputed offline as the \(\Theta\)-quantile of the net-demand random variable by numerically evaluating the CDF integral in \eqref{convolution} and solving \eqref{q_k-calculation}.


\begin{table}[h!]
\centering
\caption{Method Comparison}
\label{summary comparison}
\begin{tabular}{|l|c|c|c|}
\hline
\textbf{Feature} & \textbf{RO} & \textbf{SP} & \textbf{RC} \\
\hline
Handles uncertainty & $\checkmark$ & $\checkmark$ & $\checkmark$ \\
\hline
Low conservatism & $\times$ & $\checkmark$ & Tunable \\
\hline
Infeasibility handling & $\times$ & Implicit & $\checkmark$ \\
\hline
Good scalability & $\checkmark$ & $\times$ & $\checkmark$ \\
\hline
\end{tabular}
\end{table}

Table~\ref{summary comparison} briefly compares classical Robust Optimization (RO), Stochastic Programming (SP), and the proposed robust counterpart (RC) method parameterized by $(\delta, \gamma, \Theta)$.

\subsection{Uncertain Bilevel EV Scheduling}
\label{Uncertain}
To account for uncertainties in load demand and photovoltaic (PV) generation, the proposed robust counterpart (RC) method described in Section~\ref{methodology} is incorporated into the bilevel EV scheduling framework. Since the active and reactive power balance constraints (\ref{eq:dsopb})–(\ref{eq:dsoqb}) contain uncertain load and PV terms, they are reformulated using the RC approach as follows: 
\begin{subequations}
\label{uncertain-powerbalance}
\begin{align}
&\sum_{i\in \mathcal{I}^{b}} \! \! P^{G}_{i,t}-\!\!\!\sum_{l:s(l)=b} \!\!\!P^L_{l,t}+\!\!\!\sum_{l:r(l)=b} \!\!\!P^L_{l,t} \geq \!\!P^{f}_{b,t}- \nonumber\\ 
 &\! \gamma \max \{1,\!|(d^p_{b,t}\!-\!P^{pv}_{b,t})|\}\!+\Omega_{b,t},~~~~~~\forall b,t ~(\mu_{b,t})
 \label{eq:dsopbu} ~~~
\\[2pt]
&\Omega_{b,t}=
\begin{cases}
\!(d^p_{b,t}-P^{pv}_{b,t})+\delta \varphi \, d^p_{b,t},
& P^{pv}_{b,t}=0,
\\[6pt]
 q_{b,t},
& P^{pv}_{b,t}>0,
\end{cases}
\qquad \forall b,t
\label{eq:omega_bt}\\&
\sum_{i\in \mathcal{I}^{b}} Q^{G}_{i,t}-\sum_{l:s(l)=b} Q^L_{l,t}+\sum_{l:r(l)=b} Q^L_{l,t}  \geq Q^{f}_{b,t}+d^q_{b,t}\nonumber\\ 
&-\gamma \max \{1,|d^q_{b,t}|\}+\delta \varphi d^q_{b,t},~\forall b,t~(\rho_{b,t}),
\label{eq:dsoqbu}
\end{align}
\end{subequations}
where $\delta$ denotes the uncertainty level, $\gamma$ is the allowable infeasibility tolerance, and the relationship between $\varphi$ and the reliability level $(1-\Theta)$ is given by (\ref{reliability}). Equation~(\ref{eq:omega_bt}) shows that when PV output is zero, the uncertainty is modeled using the normal-distribution formulation in (\ref{modified}); otherwise, the normal-minus-beta formulation in (\ref{equivalent}) is employed. The quantile $q_{\Theta,k}$ is obtained from (\ref{convolution}) by solving:
\[
F_{W_k}(q_{\Theta,k}) = 1-\Theta,
\]
where $q_{\Theta,k}$ represents the upper-tail quantile of $W_k$ with exceedance probability $\Theta$.

Accordingly, the deterministic power-balance constraints (\ref{eq:dsopb})--(\ref{eq:dsoqb}) are replaced by the uncertainty-aware constraints (\ref{eq:dsopbu})--(\ref{eq:dsoqbu}). However, relaxing the active-power balance constraints from equalities to inequalities does not automatically preserve the economic interpretation of the associated dual variables $\mu_{b,t}$ as DLMPs. If the relaxed constraints are non-binding at the optimum, the corresponding dual variables may become zero or lose their marginal-price interpretation. Therefore, it is necessary to establish that the relaxed active-power balance constraints are binding at optimality. Under this property, the relaxed formulation is equivalent to the original equality-constrained model, preserving the validity and economic interpretation of the DLMPs. The following lemma establishes this result.


\begin{lemma}[Tightness of the Relaxed Active-Power Balance]
\label{lem:tightness_active}
Suppose that $P_i^{\min}=0$ for every controllable generator, including the upstream slack/source unit, and that each generation cost function $C_{i,t}(P^G_{i,t})$ is continuously differentiable and strictly increasing over its feasible domain. Then, every optimal solution of the relaxed lower-level problem satisfies the active-power balance constraints \eqref{eq:dsopbu} at equality for all $b\in\mathcal{B}$ and $t\in\mathcal{T}$.
\end{lemma}

\begin{proof}
Let
$(\boldsymbol{P}^{G*},\boldsymbol{Q}^{G*},\boldsymbol{P}^{L*},
\boldsymbol{Q}^{L*},\boldsymbol{V}^{*})$
be an optimal solution of the relaxed lower-level problem. Suppose, for contradiction, that the active-power balance constraint \eqref{eq:dsopbu} is strict for some bus $b$ and time $t$. Define the corresponding slack as
\begin{align}
s^p_{b,t}:={}&
\sum_{i\in \mathcal{I}^{b}} P^{G*}_{i,t}
-\sum_{l:s(l)=b} P^{L*}_{l,t}
+\sum_{l:r(l)=b} P^{L*}_{l,t}
\nonumber\\
&
-\Big(
P^{f}_{b,t}
+\Omega_{b,t}
-\gamma \max\{1,|(d^p_{b,t}-P^{pv}_{b,t})|\}
\Big)
>0 .
\nonumber
\end{align}

Since the constraint is nonbinding, at least one controllable generator has a positive output. Let $i^\star$ satisfy $P^{G*}_{i^\star,t}>0$, and choose:
\[
0<\epsilon\le
\min\left\{s^p_{b,t},\,P^{G*}_{i^\star,t}\right\}.
\]
Construct a perturbed solution by setting:
\[
P^{G\prime}_{i^\star,t}
=
P^{G*}_{i^\star,t}-\epsilon,
\qquad
P^{G\prime}_{i,t}=P^{G*}_{i,t},
\ \forall i\neq i^\star,
\]
while leaving all other variables unchanged. 

Since $\epsilon\le s^p_{b,t}$, the relaxed power-balance constraint \eqref{eq:dsopbu} remains feasible. Moreover, $\epsilon\le P^{G*}_{i^\star,t}$ implies:
\[
P^{G\prime}_{i^\star,t}
\ge
P^{\min}_{i^\star}=0,
\]
so all generator limits remain satisfied. As no other variables are modified, all remaining constraints remain feasible.

Finally, because $C_{i^\star,t}(\cdot)$ is strictly increasing:
\[
C_{i^\star,t}(P^{G*}_{i^\star,t}-\epsilon)
<
C_{i^\star,t}(P^{G*}_{i^\star,t}),
\]
which strictly decreases the objective value, contradicting the optimality of
$(\boldsymbol{P}^{G*},\boldsymbol{Q}^{G*},\boldsymbol{P}^{L*},
\boldsymbol{Q}^{L*},\boldsymbol{V}^{*})$.
Therefore, constraint \eqref{eq:dsopbu} must be binding at every optimal solution.
\end{proof}


For the reactive-power balance constraints, an arbitrarily small regularization term may be added to the lower-level objective to penalize unnecessary reactive-power provision:
\[
\varepsilon \sum_{b,t}
(
\sum_{i\in\mathcal{I}^{b}} Q^G_{i,t}
+
\sum_{e\in\mathcal{E}^{b}} |Q^E_{e,b,t}|
),
\qquad
0<\varepsilon\ll1,
\]
which eliminates degenerate solutions with excess reactive-power injection. Since $\varepsilon$ is infinitesimal relative to the active-power generation cost, this regularization does not affect the optimal active-power dispatch, DLMPs, or other primal decisions. Consequently, every optimal solution minimizes unnecessary reactive-power supply, implying that the relaxed reactive-power balance constraints are also binding at optimality. Therefore, both the active- and reactive-power balance constraints can be represented as equalities, yielding:
\begin{subequations}
\label{uncertain-powerbalance-e}
\begin{align}
&\sum_{i\in \mathcal{I}^b} \! \! P^{G}_{i,t}-\!\!\!\sum_{l:s(l)=b} \!\!\!P^L_{l,t}+\!\!\!\sum_{l:r(l)=b} \!\!\!P^L_{l,t} = \!\!P^{f}_{b,t}- \nonumber\\ 
 &\! \gamma \max \{1,\!|(d^p_{b,t}\!-\!P^{pv}_{b,t})|\}\!+\Omega_{b,t},~~~~~~\forall b,t ~(\mu_{b,t})
 \label{eq:dsopbue}\\
&\sum_{i\in \mathcal{I}^b} Q^{G}_{i,t}-\sum_{l:s(l)=b} Q^L_{l,t}+\sum_{l:r(l)=b} Q^L_{l,t} 
 = Q^{f}_{b,t}+d^q_{b,t}\nonumber\\ &-\gamma \max \{1,|d^q_{b,t}|\}+\delta \varphi d^q_{b,t},~\forall b,t~(\rho_{b,t}).
\label{eq:dsoqbue}
\end{align}
\end{subequations}

To incorporate uncertainty, the upper-level EV aggregator problem remains unchanged, whereas the deterministic lower-level EMS problem is replaced by the uncertainty-aware formulation developed in Section~\ref{methodology}. This formulation accounts for uncertainty in the relevant system parameters while preserving the network and operational constraints of the original model. The resulting bilevel problem is then transformed into an equivalent single-level formulation by applying the KKT conditions described in Section~\ref{single level}. Accordingly, the uncertainty-aware single-level optimization model is formulated as follows:
\begin{equation}
\label{uncertain single-level-reformulation}
\Delta^u:=\bigg\{(\ref{DET-upper}), (\ref{eq:dsov})-(\ref{eq:dsovl}), (\ref{KKT}), \eqref{eq:omega_bt}, (\ref{eq:dsopbue})-(\ref{eq:dsoqbue})\bigg\}.
\end{equation}

The nonlinear complementary slackness conditions (\ref{dc6})--(\ref{dc15}) are linearized using the Big-$M$ technique. Specifically, a bilinear complementarity condition of the form $xz=0$ is reformulated as:
\begin{subequations}
\label{big-M}
\begin{align}
& x \le M y, ~~z \le M (1 - y), \\
& x \ge 0,\; z \ge 0, ~ y \in \{0,1\},
\end{align}
\end{subequations}
where \(M\) denotes a sufficiently large positive constant. Based on (\ref{big-M}), if \(y=1\), then \(z=0\); otherwise, \(x=0\). Therefore, the complementarity condition \(xz=0\) is enforced.

\section{Simulation Results}
\label{Results}
The simulation horizon covers 24 hours with 15-minute intervals. The proposed framework is intended for supervisory day-ahead and intra-day scheduling rather than fast real-time control. All simulations are implemented in Python~\cite{Python} and solved using Gurobi. The framework is evaluated on the IEEE 33-bus distribution test system using the network data in~\cite{santoso2018optimal}. Bus voltage magnitudes are maintained within 0.95--1.05 p.u. on a 12.66-kV base. Line thermal limits follow the test system specifications, with active and reactive power limits ranging from 1 to 6 MW and 0.8 to 4 MVAR, respectively.

Two DGs are located at buses~1 and~8, where bus~1 represents the upstream grid connection. Generation cost functions are adopted from~\cite{golshannavaz2014smart}, with maximum active/reactive capacities of 7~MW/4~MVAR for the substation and 4~MW/3~MVAR for the DG at bus~8. The quadratic generation-cost coefficient is set to zero because its contribution is negligible over the considered operating range.
 Active and reactive load demands are generated from uniform distributions $U[56,390]$~kW and $U[22,220]$~kVAR, respectively, representing typical commercial feeder conditions. EV arrival time, departure time, initial SOC, and desired SOC are generated from truncated Gaussian distributions~\cite{data1} with parameters $\mathcal{N}(48,26)$, $\mathcal{N}(68,20)$, $\mathcal{N}(0.4,0.1)$, and $\mathcal{N}(0.6,0.1)$, respectively. Each EV charger has a 12~kVA rating with 90\% charging efficiency, and the SOC is maintained between 20\% and 80\% of battery capacity. Two EV fleets, each consisting of 100 EVs with 30~kWh batteries, are located at buses~18 and~33. Solar PV units are installed at buses~11, 25, and~33, with generation profiles adopted from~\cite{zhu2016graphical} and scaled to match the operating conditions of the test system.

\subsection{Deterministic EV Scheduling}
\label{Deterministic EV Scheduling}
In this case, the bilevel model in Section~\ref{deterministic} is solved deterministically. EVs operate at non-unity power factor, enabling simultaneous active charging and reactive power support. Fig.~\ref{LMPvsperiod_DET} shows the resulting DLMPs at buses~18 and~33 over the 96 scheduling intervals. The DLMPs correspond to the dual variables $\mu_{b,t}$ of the active-power balance constraints obtained from the single-level reformulation. Their temporal variation is primarily driven by network congestion and voltage constraints. For example, at period~78, the DLMP at bus~33 is \$0.085/kWh, whereas the corresponding value at bus~18 is \$0.115/kWh, illustrating the location-dependent nature of marginal electricity prices.

Fig.~\ref{powervsperiod_DET} illustrates the active and reactive power profiles of the EV aggregator. At bus~18 (Fig.~\ref{powervsperiod_DET_17}), EVs draw active power for charging while injecting reactive power, operating in the fourth quadrant of the $P$--$Q$ plane. For instance, at period~93, the EV fleet draws 111~kW and injects 179~kVAR. Similarly, the fleet at bus~33 (Fig.~\ref{powervsperiod_DET_32}) draws 114~kW and injects 177~kVAR. The injected reactive power provides local voltage support by mitigating voltage drops along the feeder. The reported reactive power corresponds to the aggregated capability of the EV fleet at each bus.
Fig.~\ref{voltagevsperiod_DET} compares the voltage profiles at buses~18 and~33 under unity and non-unity power factor operation. As shown in Figs.~\ref{voltagevsperiod_DET_17} and~\ref{voltagevsperiod_DET_32}, the voltages remain within the allowable range of 0.95--1.05~p.u. in both cases. However, because system uncertainties are neglected, the deterministic model may yield overly optimistic operating conditions and fail to capture voltage violations that could arise in practical operation.


\vspace{-0.3cm}

\begin{figure}[h!]
	\centering
		\includegraphics[width=0.34\textwidth,height=0.12\textheight]{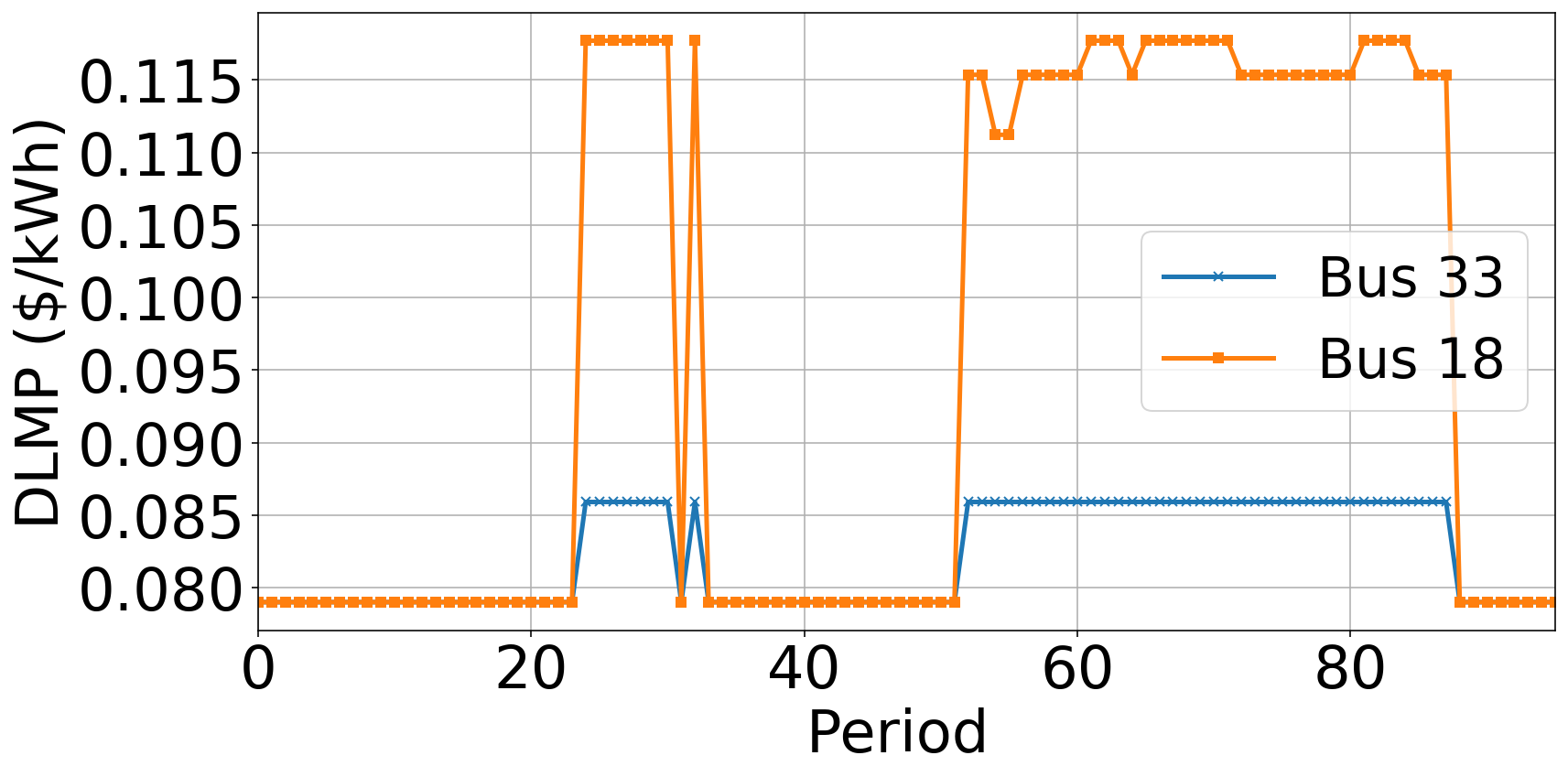}
			\caption{Temporal variation of DLMPs at buses 18 and 33}
	\label{LMPvsperiod_DET}
\end{figure}

\begin{figure}[h!]
\vspace{-0.6 cm}
\centering
        \subfigure[Flexible power at bus 18]{
	     \includegraphics[width=0.23\textwidth,height=0.11\textheight]{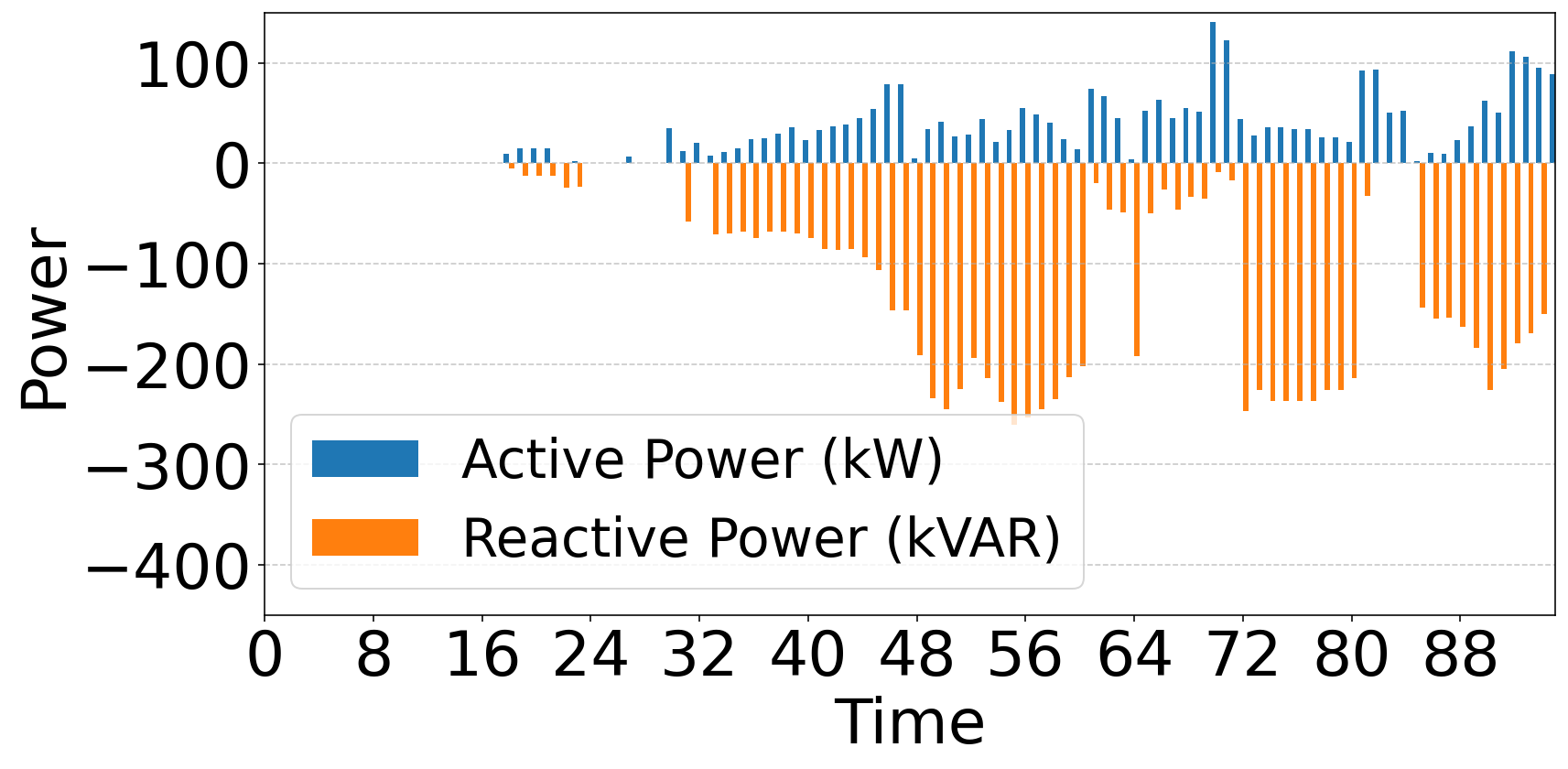}
	     \label{powervsperiod_DET_17}
	}  \hspace*{-0.58 cm} 
	     \subfigure[Flexible power at bus 33]{
	     \includegraphics[width=0.23\textwidth,height=0.11\textheight]{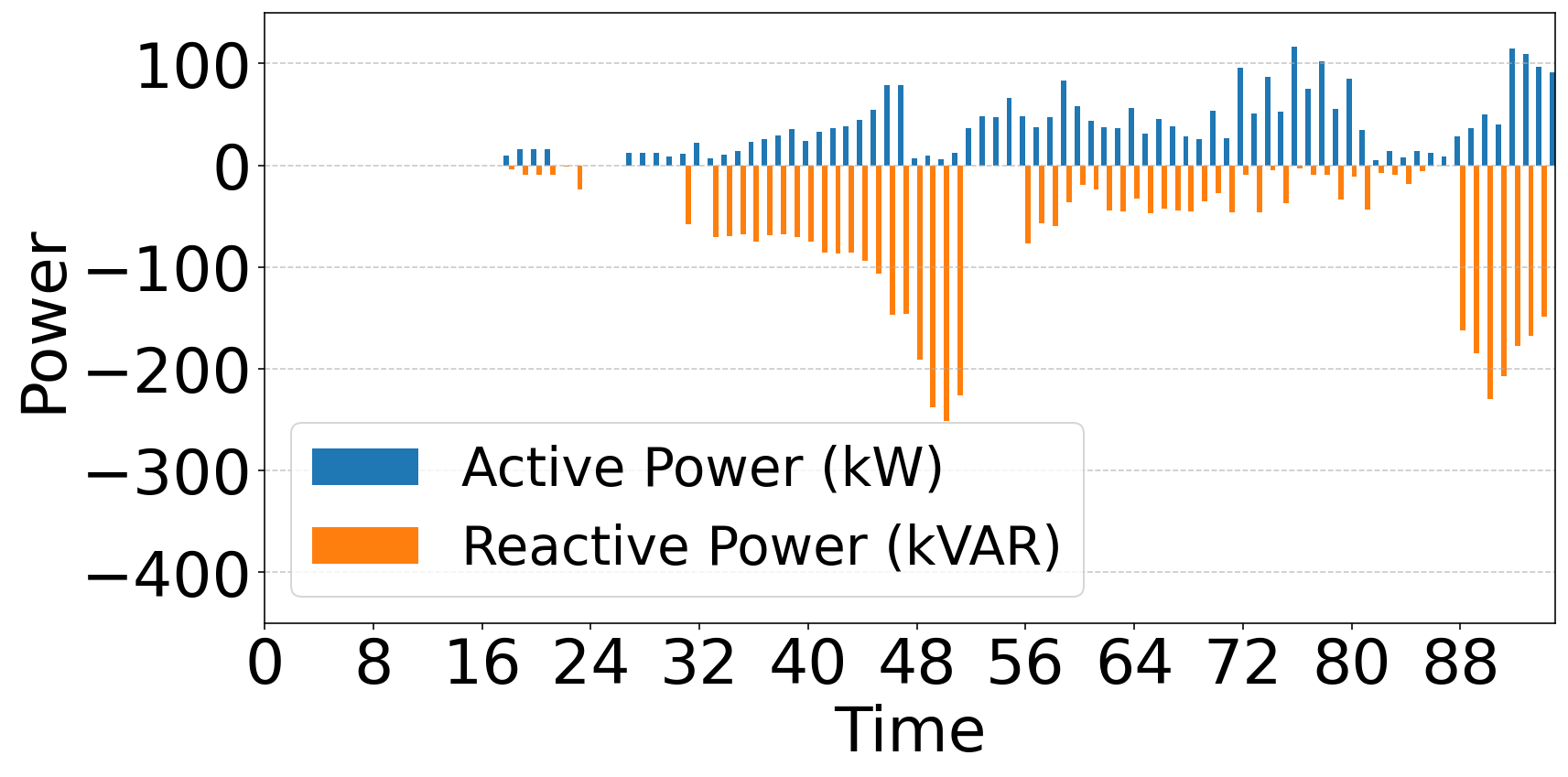}
	     \label{powervsperiod_DET_32}
	} \vspace{- 0.0cm}
			\caption{Flexible power for EVs at buses 18 and 33}
                \label{powervsperiod_DET}
                 \vspace{ 0.0cm}
\end{figure}

\begin{figure}[h!]
\vspace{-0.6 cm}
\centering
        \subfigure[Voltage profile at bus 18]{
	     \includegraphics[width=0.23\textwidth,height=0.11\textheight]{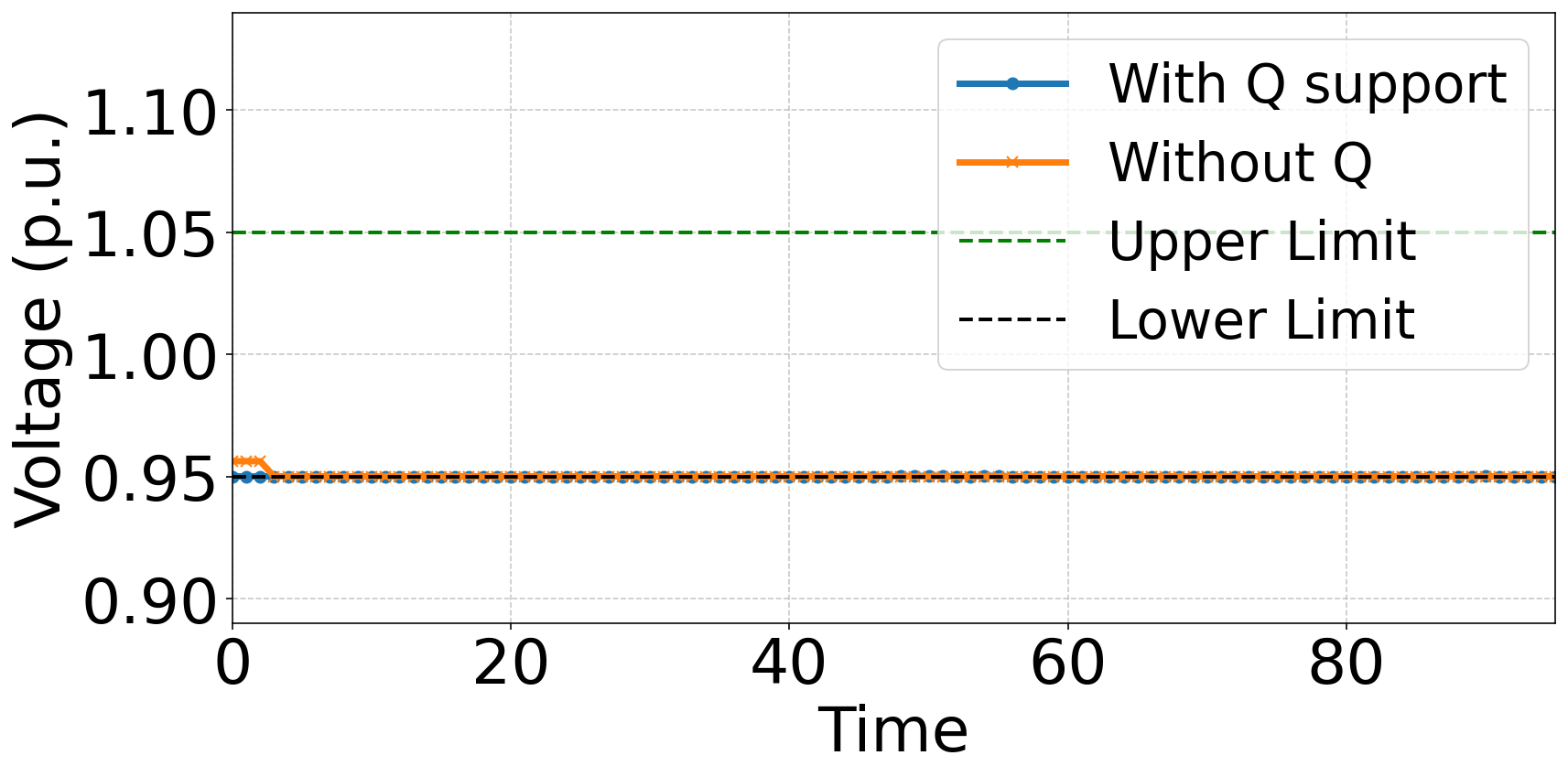}
	     \label{voltagevsperiod_DET_17}
	}  \hspace*{-0.53 cm} 
	     \subfigure[Voltage profile at bus 33]{
	     \includegraphics[width=0.23\textwidth,height=0.11\textheight]{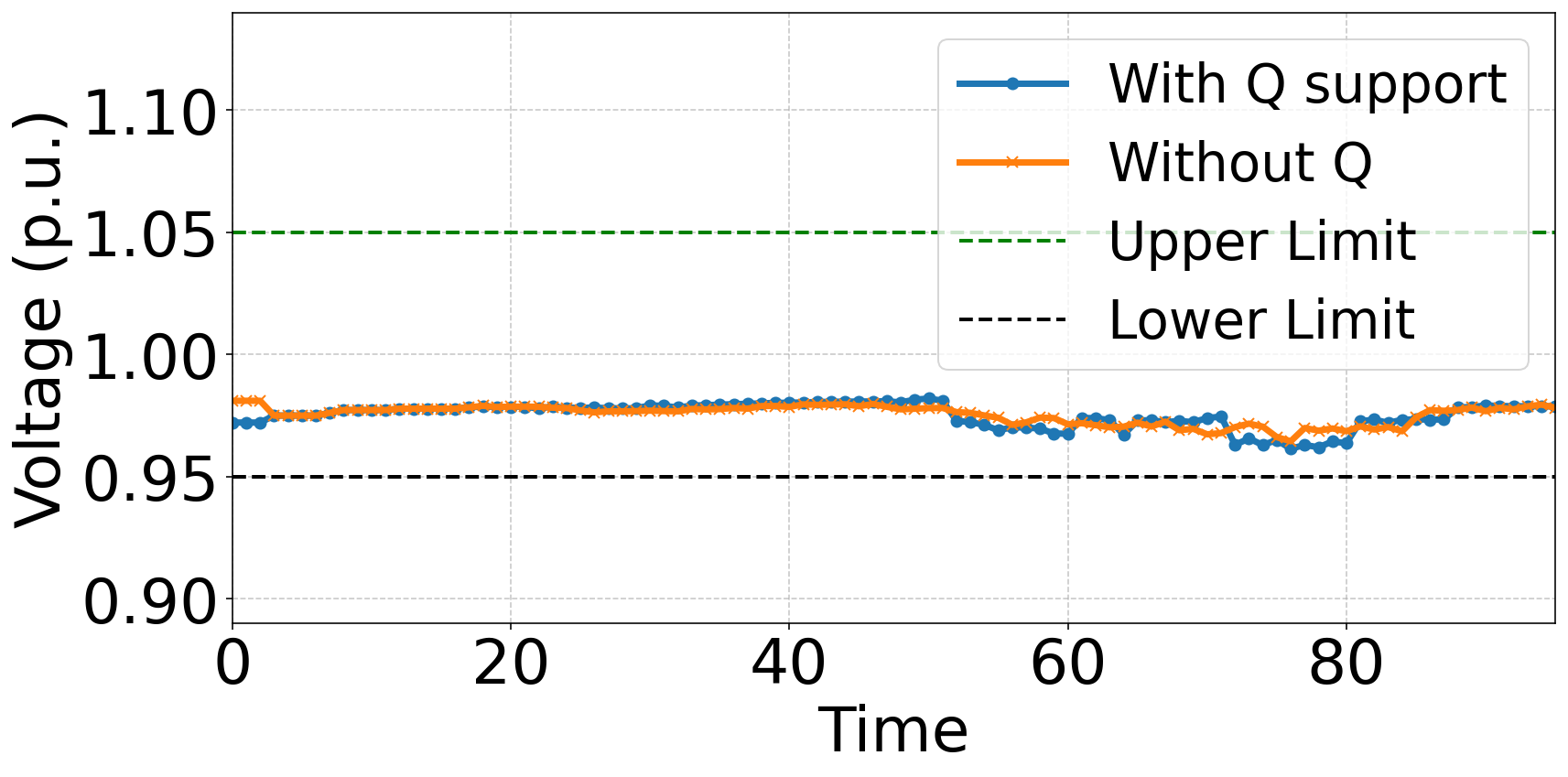}
	     \label{voltagevsperiod_DET_32}
	} \vspace{- 0.0cm}
			\caption{Voltage profiles at buses 18 and 33 }
                \label{voltagevsperiod_DET}
                 \vspace{- 0.5cm}
\end{figure} 

\subsection{Uncertain EV Scheduling}
\label{Uncertain EV Scheduling}
In this case, the uncertainty-aware framework described in Section~\ref{Uncertain} is applied to evaluate the impact of load and PV generation uncertainty on EV scheduling. The uncertainty parameters are set to $\delta=0.18$, $\varphi=1.5$ (corresponding to a 93\% reliability level), and $\gamma=0$; the same reliability level is used in \eqref{q_k-calculation}. Setting $\gamma=0$ enforces strict feasibility, although the framework readily accommodates controlled infeasibility relaxation when desired by system operators. Fig.~\ref{LMPvsperiod_Unc} shows the resulting DLMPs at buses~18 and~33. Compared with the deterministic case, the DLMPs increase due to the additional conservatism required to accommodate load and PV uncertainty. In general, larger values of $\delta$ and $\varphi$ lead to higher DLMPs and greater reactive power support. At period~46, the DLMP increases from \$0.079/kWh to \$0.085/kWh at bus~33 and from \$0.079/kWh to \$0.117/kWh at bus~18. 


Fig.~\ref{powervsperiod_Unc} illustrates the active and reactive power behavior of EVs under uncertainty. At period~81, the EV fleet at bus~18 draws 64~kW of active power while injecting 107~kVAR of reactive power (Fig.~\ref{powervsperiod_Unc_17}), whereas the EVs at bus~33 draw 55~kW and inject 62~kVAR (Fig.~\ref{powervsperiod_Unc_32}).  
These results demonstrate that uncertainty-aware scheduling leads to more conservative operational decisions, characterized by elevated DLMPs and required reactive power contributions from EVs. These adjustments improve voltage regulation and enhance the operational resilience of the distribution feeder under load and PV generation uncertainty. All EV departure SOC requirements remained satisfied throughout the scheduling horizon despite the provision of reactive power support.

Fig.~\ref{voltagevsperiod_Unc} compares the voltage profiles at buses~18 and~33 under unity and non-unity power factor operation. At bus~18, unity-power-factor charging provides no reactive power support, causing the voltage to fall below the 0.95~p.u. limit and reach approximately 0.90~p.u. during several periods. In contrast, reactive power injection through fourth-quadrant operation maintains the voltage within the permissible range, as shown in Fig.~\ref{voltagevsperiod_Unc_17}. A similar trend is observed at bus~33 (Fig.~\ref{voltagevsperiod_Unc_32}), where reactive power support mitigates voltage drops and maintains the voltage within the 0.95--1.05~p.u. limits. Furthermore, incorporating uncertainty increases the objective value by approximately 10\%. 



These results demonstrate that reactive power support from EVs significantly improves voltage regulation under uncertainty. In the deterministic case, voltage magnitudes remain within the allowable range even without reactive support, potentially leading to overly optimistic operating conditions. However, when load and PV uncertainties are considered, significant voltage drops occur, highlighting the practical need for grid-supportive EV operation.

From an operational perspective, uncertainty-aware scheduling encourages EV aggregators to provide greater reactive power support during periods of system stress. This coordinated response mitigates voltage deviations, enhances feeder resilience, and reduces reliance on costly network reinforcements or dedicated voltage control devices, demonstrating the potential of EV fleets as cost-effective flexibility resources for future distribution systems.

\subsection{Model Comparisons}
To evaluate the computational tractability of the proposed reformulation, conventional scenario-based stochastic programming (SP) and uncertainty-set-based robust optimization (RO) models were also implemented. After KKT reformulation and Big-$M$ linearization, all uncertainty-aware bilevel formulations become large-scale MPEC/MINLP problems. In the SP formulation, each scenario introduces additional primal, dual, and complementarity variables, whereas the RO formulation requires additional uncertainty-set constraints and dual variables. For a fair comparison, all models were implemented using the same solver environment, parameter settings, and Big-$M$ values.

\begin{figure}[h!]
	\centering
		\includegraphics[width=0.34\textwidth,height=0.12\textheight]{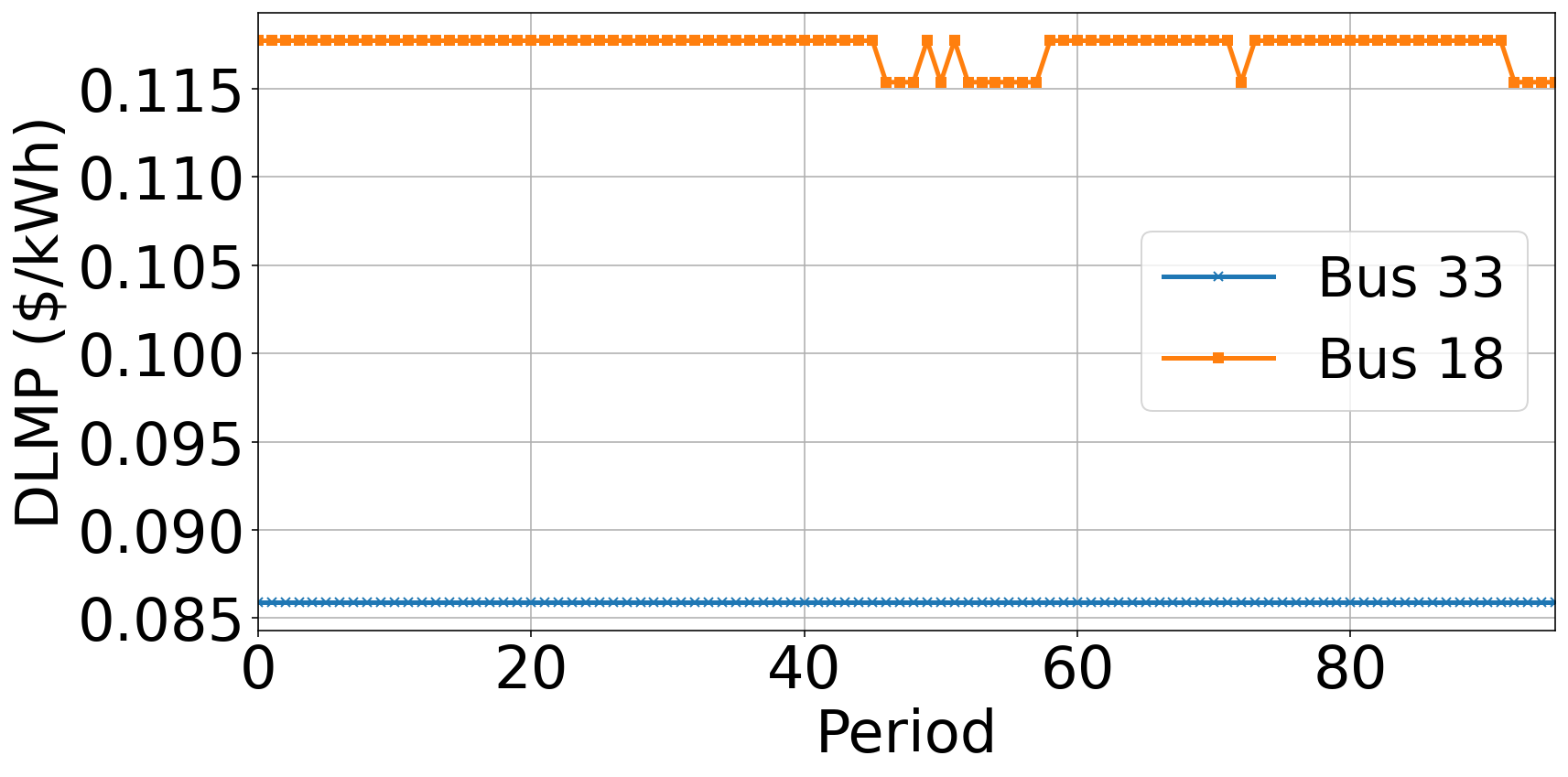}
			\caption{DLMP profiles at buses 18 and 33 under uncertainty}
	    \label{LMPvsperiod_Unc}
	\end{figure}

\vspace{-0.8cm}

\begin{figure}[h!]
\vspace{0.0 cm}
\centering
        \subfigure[Flexible power at bus 18]{
	     \includegraphics[width=0.23\textwidth,height=0.11\textheight]{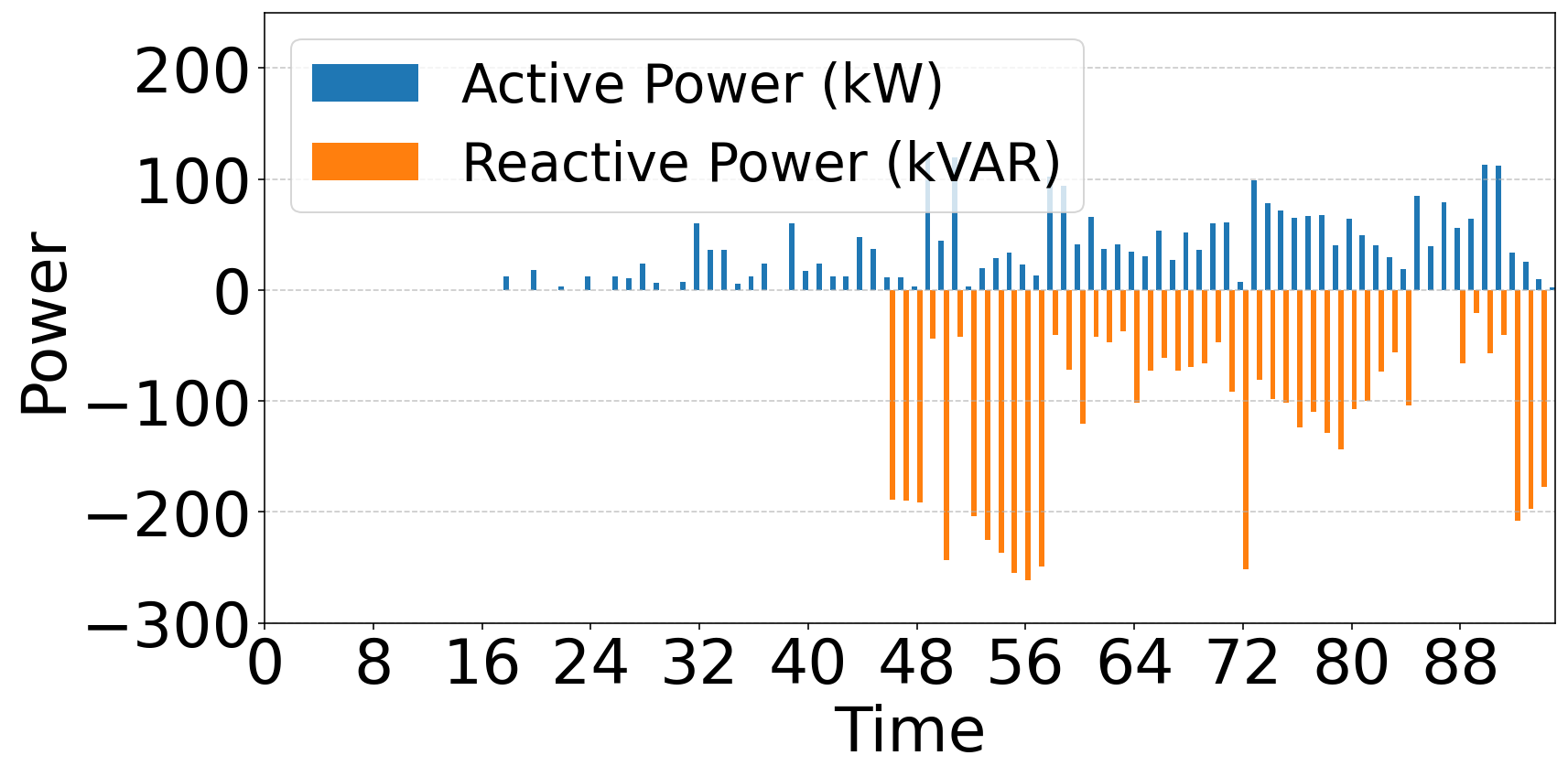}
	     \label{powervsperiod_Unc_17}
	}  \hspace*{-0.58 cm} 
	     \subfigure[Flexible power at bus 33]{
	     \includegraphics[width=0.23\textwidth,height=0.11\textheight]{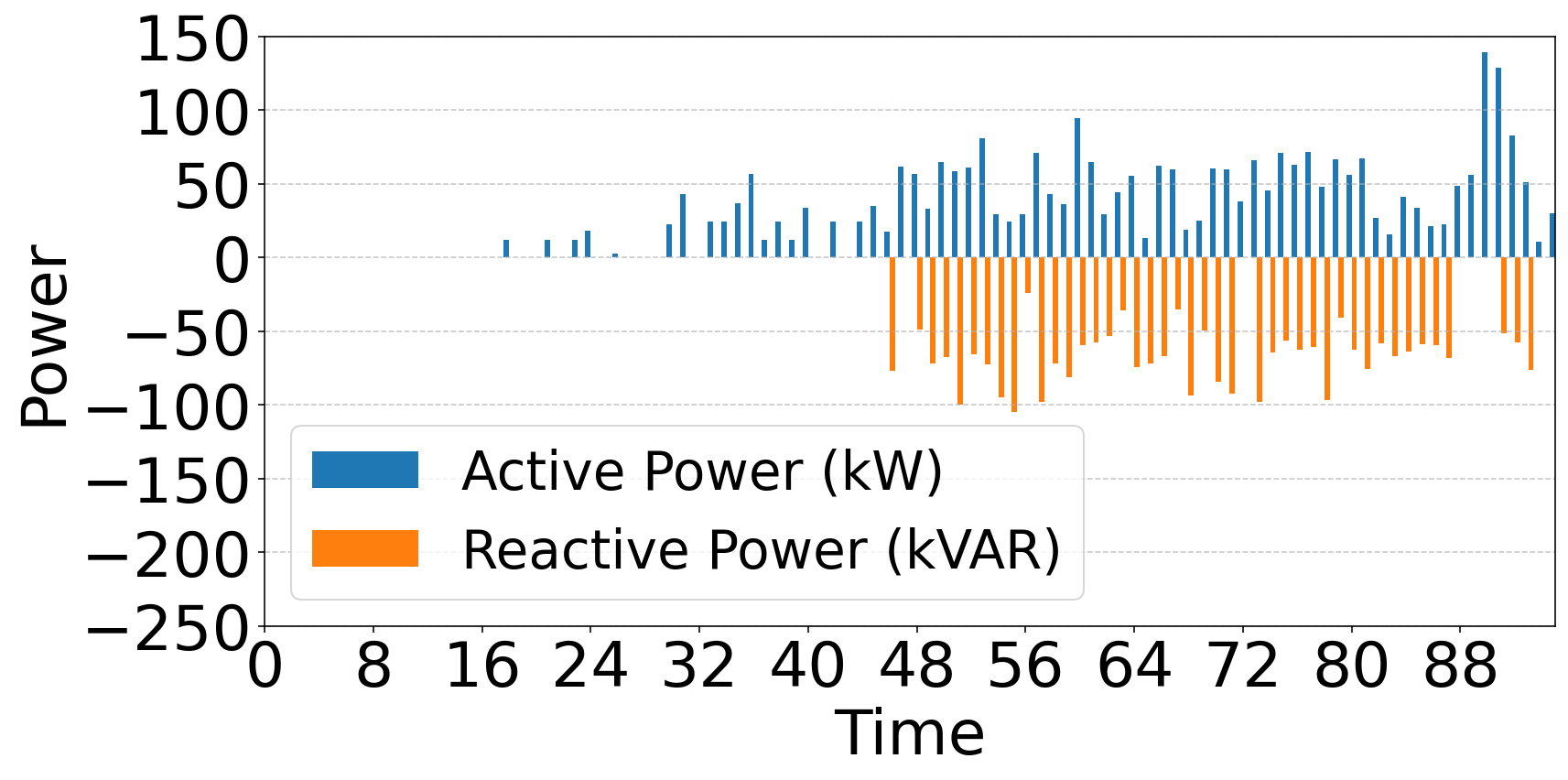}
	     \label{powervsperiod_Unc_32}
	} \vspace{- 0.0cm}
			\caption{EV power at buses 18 and 33 under uncertainty}
                \label{powervsperiod_Unc}
                 \vspace{ 0.0cm}
\end{figure}
\vspace{ -0.8cm}
\begin{figure}[h!]
\vspace{0.0 cm}
\centering
        \subfigure[Voltage profile at bus 18]{
	     \includegraphics[width=0.23\textwidth,height=0.11\textheight]{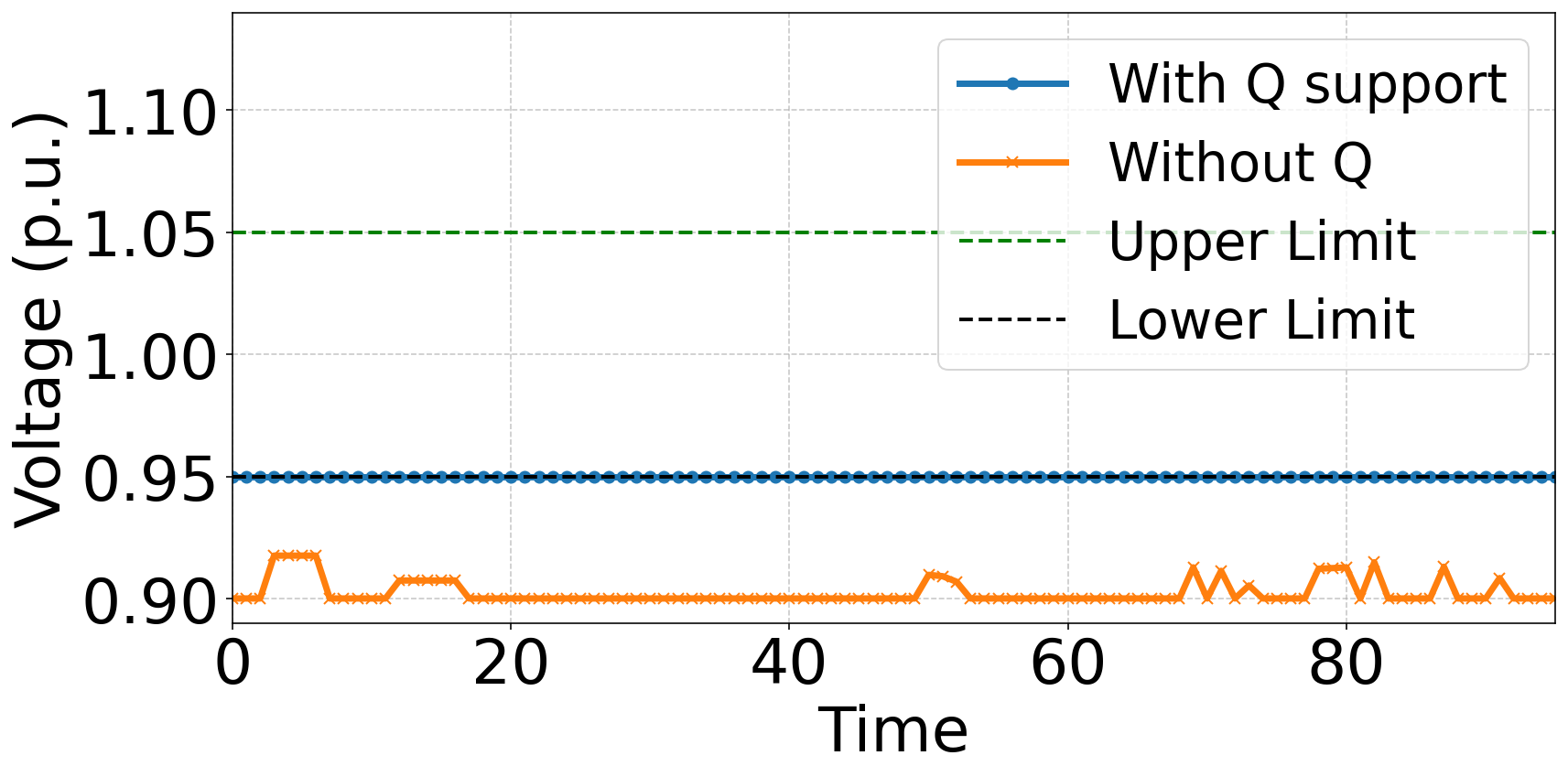}
	     \label{voltagevsperiod_Unc_17}
	}  \hspace*{-0.53 cm} 
	     \subfigure[Voltage profile at bus 33]{
	     \includegraphics[width=0.23\textwidth,height=0.11\textheight]{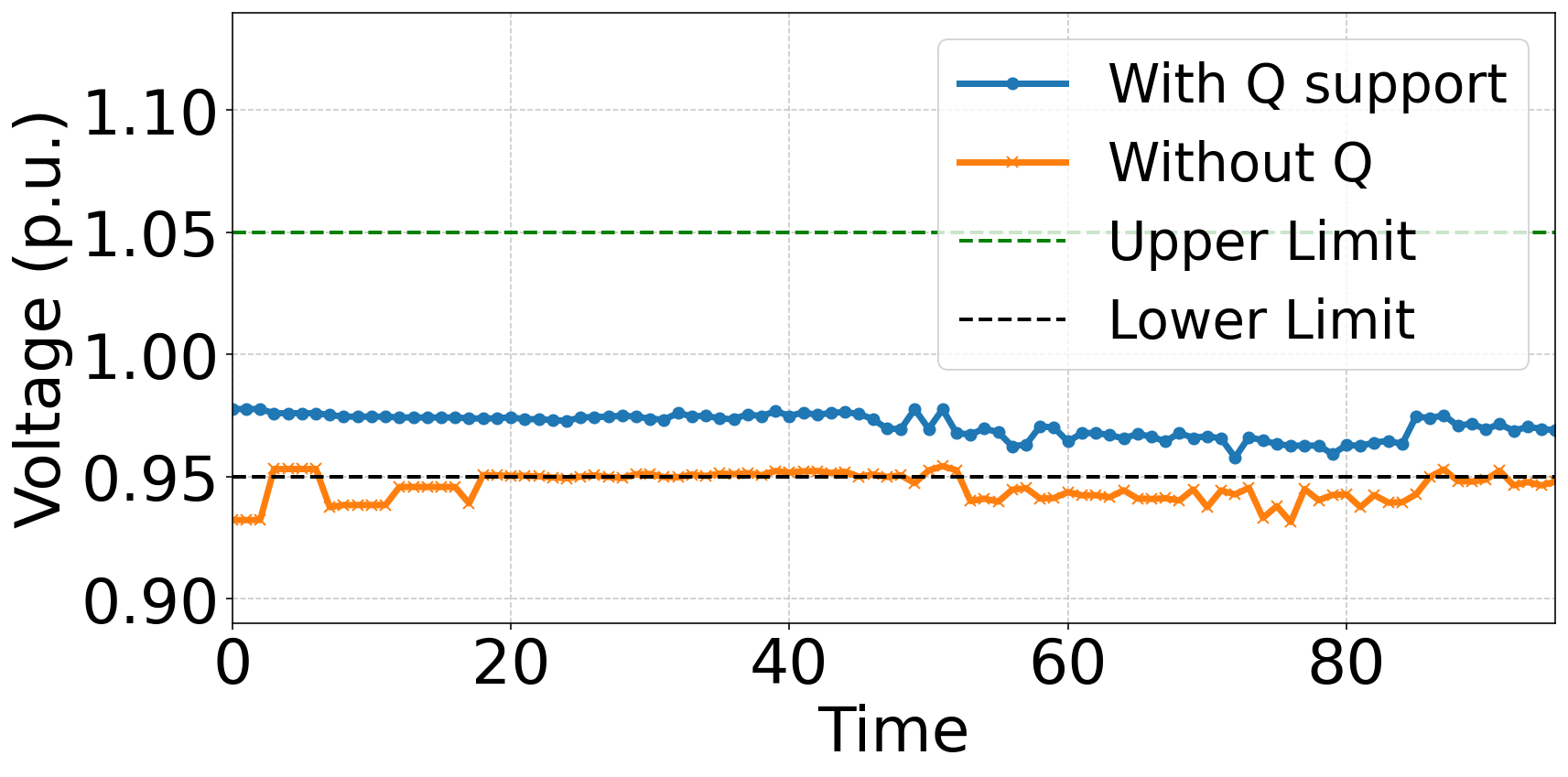}
	     \label{voltagevsperiod_Unc_32}
	} \vspace{- 0.0cm}
			\caption{Voltage magnitude across different time periods}
                \label{voltagevsperiod_Unc}
                 \vspace{- 0.0cm}
\end{figure}


\begin{table}[h!]
\caption{Computational Comparison Under Different Big-M Values}
\label{runtime}
\centering
\begin{tabular}{c c c c}
\hline
Method & Big-M Value & Runtime & Optimality Gap \\
\hline
SP & $M=100$  & $>2$ h   & Not converged \\
SP & $M=500$  & $>2$ h   & Not converged \\
SP & $M=1000$ & $>2$ h   & Not converged \\
\hline
RO & $M=100$  & 708 s   & 0.6\% \\
RO & $M=500$  & $>2$ h   & Not converged \\
RO & $M=1000$ & $>2$ h   & Not converged \\
\hline
\textbf{RC} & $M=100$  & 4 s  & 0.5\% \\
\textbf{RC} & $M=500$  & 10 s  & 0.5\% \\
\textbf{RC} & $M=1000$ & 12 s & 0.5\% \\
\hline
\end{tabular}
\end{table}
For the SP benchmark, 50 equiprobable scenarios were generated, where demand follows a truncated normal distribution within $\pm10\%$ of its forecast value and PV generation follows a scaled beta distribution within $\pm15\%$, with zero output during non-generating periods. For the RO benchmark, the uncertainty budget was set to 20. The nominal demand was defined as 25\% below the forecast value, and the maximum upward deviation was equal to the nominal demand.

For the full 96-interval scheduling horizon, the SP and RO formulations exhibited severe scalability issues and failed to converge to acceptable optimality gaps within practical solution times. In contrast, the proposed RC reformulation preserved a compact deterministic structure and achieved tractable solution times while maintaining uncertainty-aware DLMP consistency, as summarized in Table~\ref{runtime}. Representative SP and RO formulations are omitted for brevity.



To enable direct numerical comparison, the scheduling horizon was reduced to 24 hourly intervals, allowing the SP and RO benchmarks to converge. Fig.~\ref{fig:model_comparison} compares the resulting DLMP and voltage profiles. The proposed RC formulation yields higher DLMPs at buses~18 and~33 than the SP model, corresponding to an objective value of 151 compared with 137 for SP. In contrast, the RO model produces lower DLMPs and an objective value of 127, although this does not imply lower conservatism. Under the adopted RO uncertainty set, the nominal demand is defined as the minimum demand level, and only selected buses deviate upward according to the uncertainty budget. Consequently, charging-station buses may remain close to the nominal demand in the worst-case realization, reducing local congestion and voltage binding effects. This result highlights that system-level conservatism does not necessarily translate into higher DLMPs at every bus, as local prices depend on congestion and voltage conditions.

In all three models, EV reactive power support maintains the voltage profiles at the charging-station buses within the permissible limits, as shown in Fig.~\ref{fig:model_comparison}.
\begin{figure}[h!]
\centering
 \subfigure[DLMP profiles for SP]{
	     \includegraphics[width=0.225\textwidth,height=0.11\textheight]{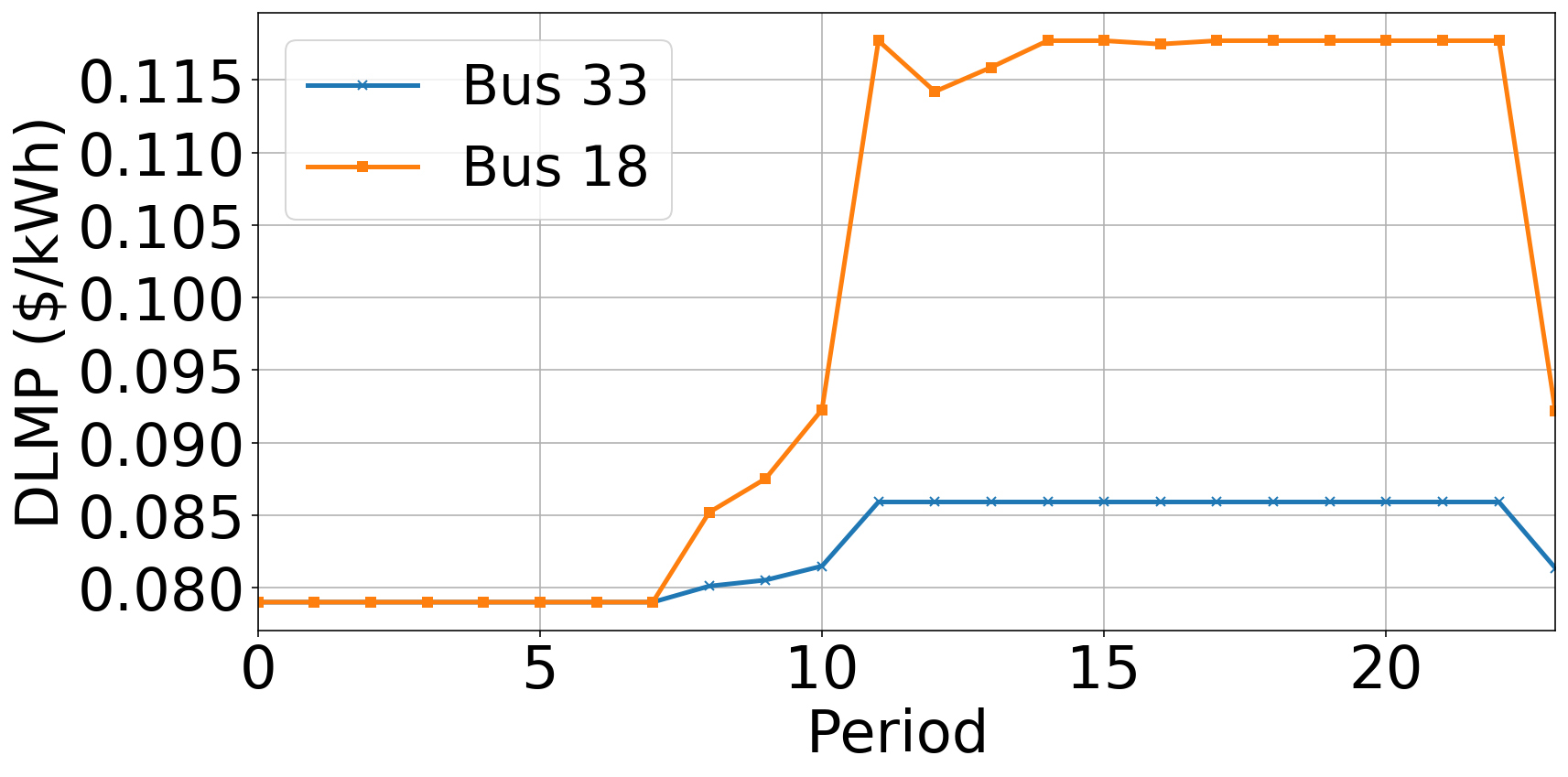}
	     \label{fig:LMP_SP}
	}   \hspace*{-0.3 cm} 
	     \subfigure[Voltage profiles for SP]{
	     \includegraphics[width=0.237\textwidth,height=0.11\textheight]{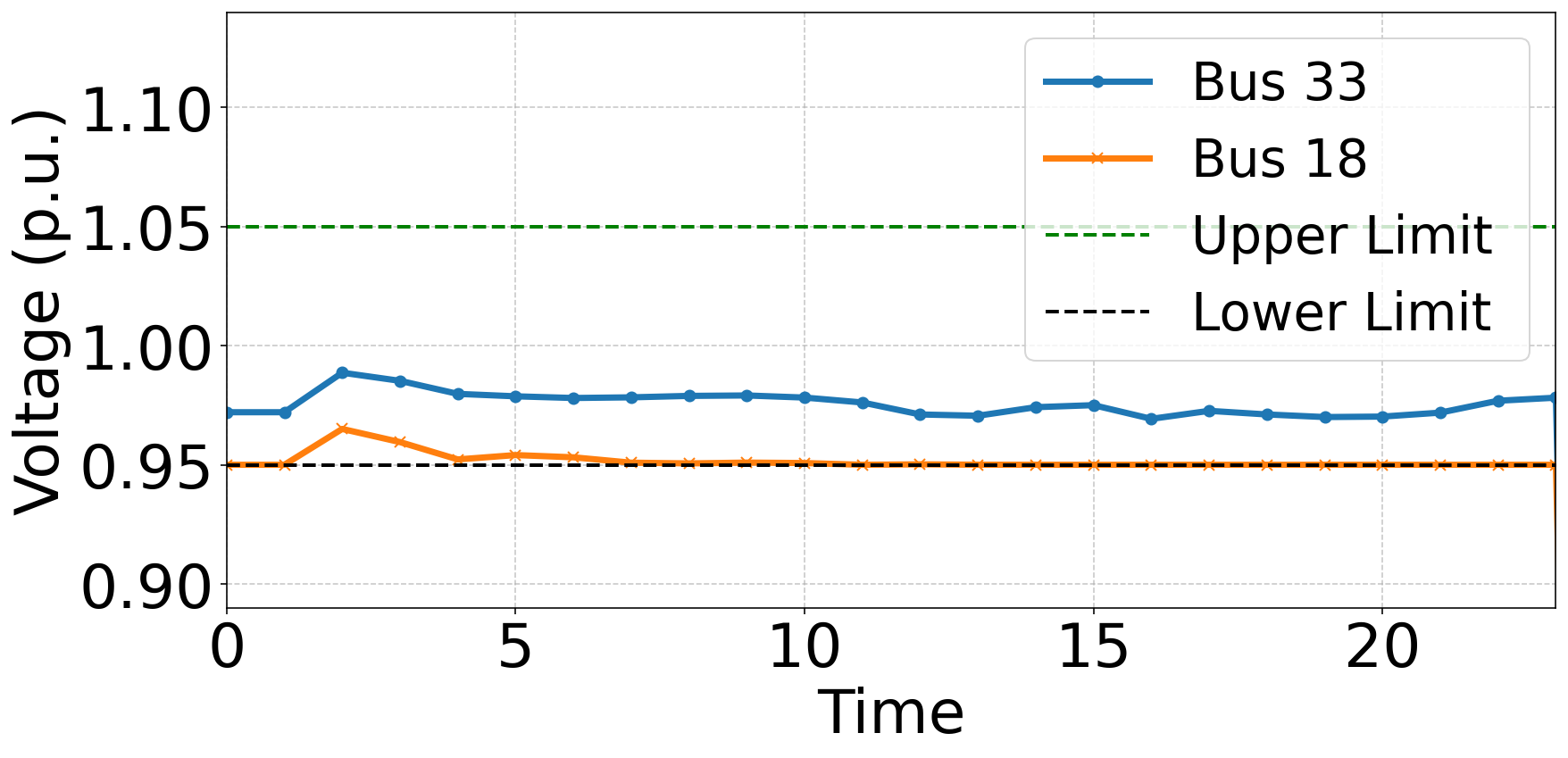}
	     \label{fig:voltagevsperiod_SP}
	} \vspace{- 0.1cm}
        \subfigure[DLMP profiles for RO]{
	     \includegraphics[width=0.237\textwidth,height=0.11\textheight]{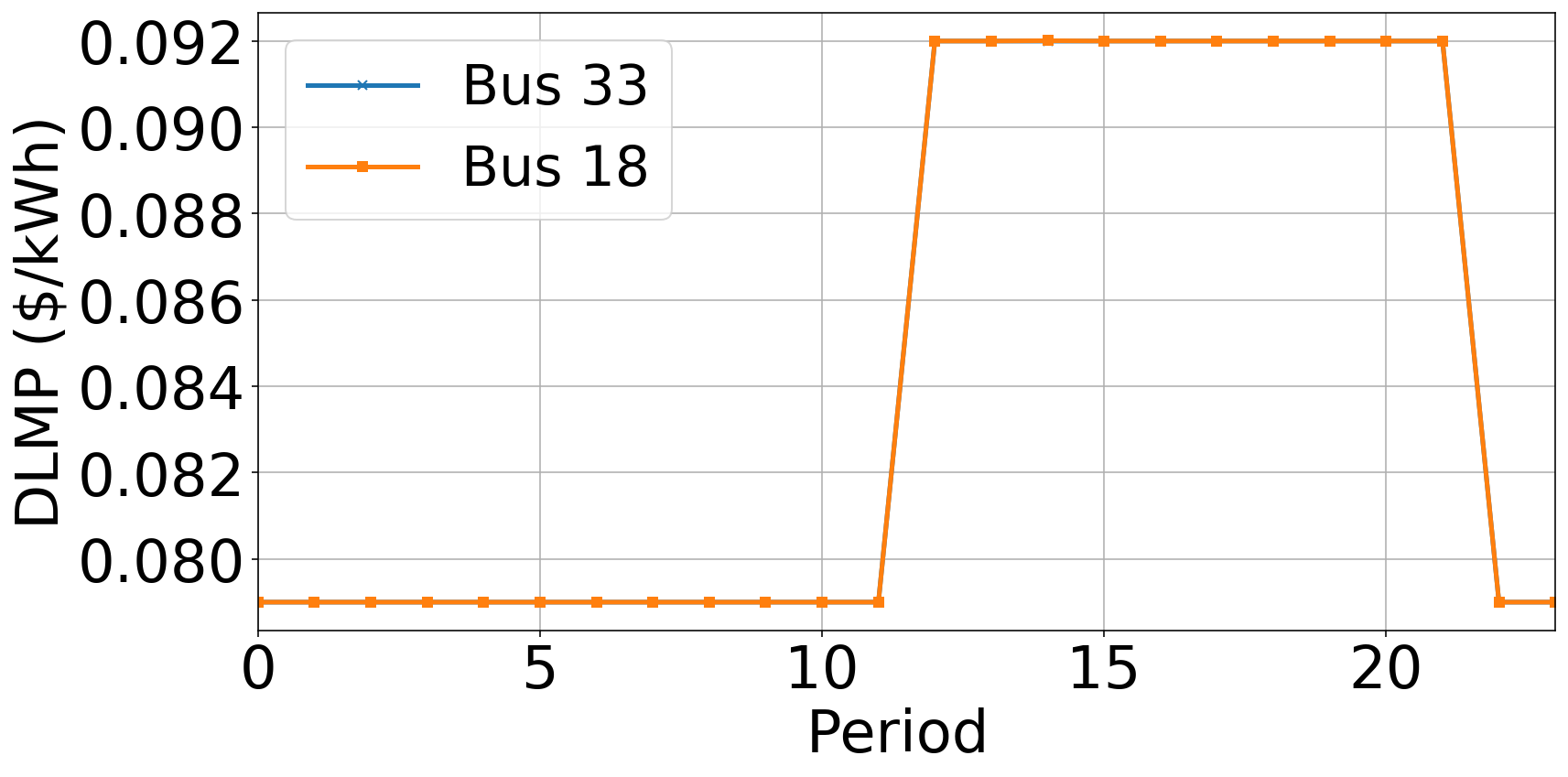}
	     \label{fig:LMPvsPriod_RO}
	} \hspace*{-0.5 cm} 
	     \subfigure[Voltage profiles for RO]{
	     \includegraphics[width=0.237\textwidth,height=0.11\textheight]{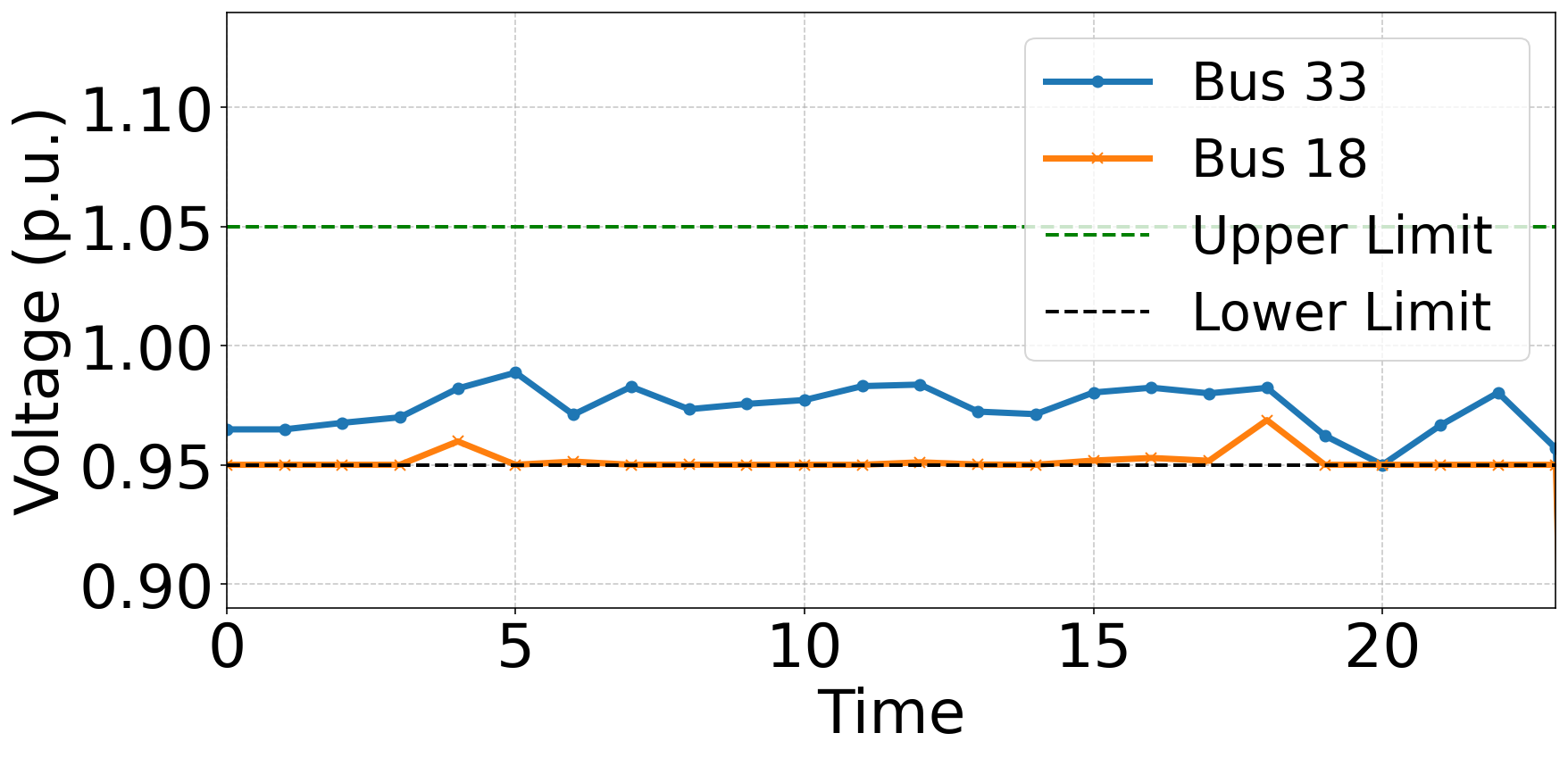}
	     \label{fig:voltagevsperiod_RO}
	} \vspace{- 0.1cm}
        \subfigure[DLMP profiles for RC]{
	     \includegraphics[width=0.242\textwidth,height=0.115\textheight]{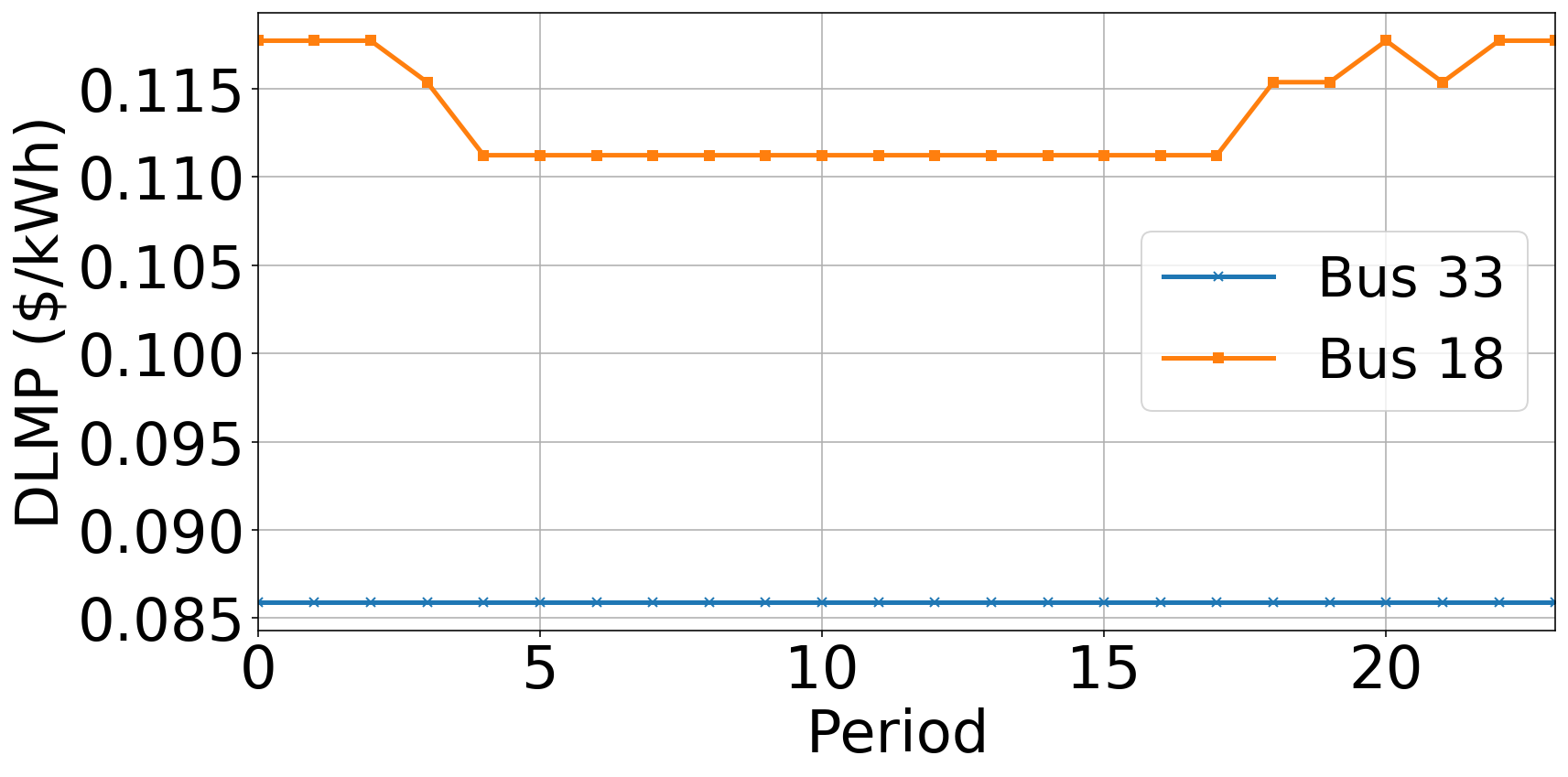}
	     \label{fig:LMPvsPriod_RC}
	}  \hspace*{-0.6 cm}  
	     \subfigure[Voltage profiles for RC]{
	     \includegraphics[width=0.235\textwidth,height=0.115\textheight]{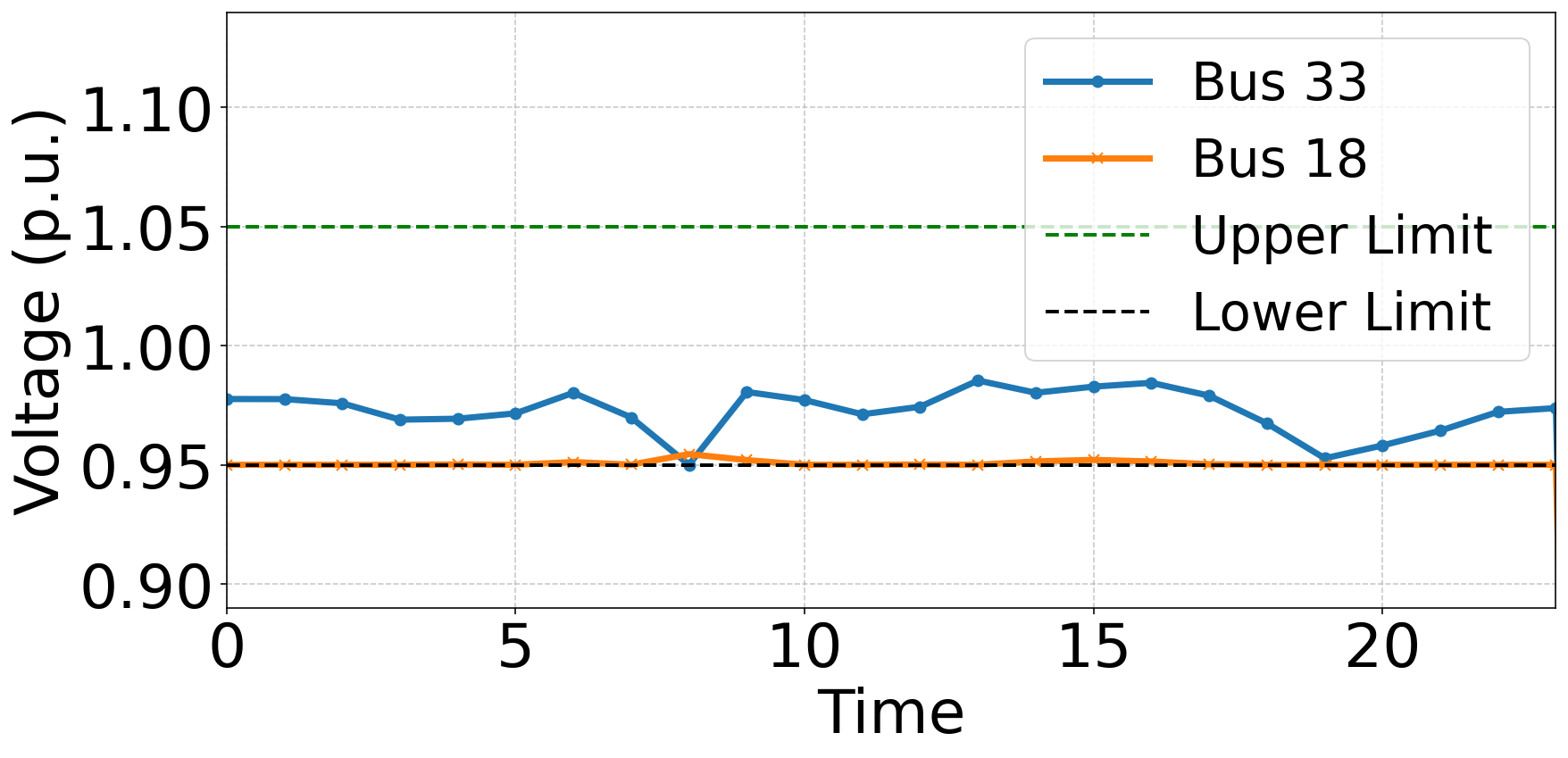}
	     \label{fig:voltagevsperiod_RC_17-32}
	}\vspace{-0.0cm}
	    \caption{DLMP and voltage profiles under SP, RO, and RC}
                \label{fig:model_comparison}
\end{figure}

\section{Conclusion}
\label{Conclusion}
This paper presented a scenario-free uncertainty-aware bilevel optimization framework for coordinated EV charging and reactive power support in distribution networks using DLMPs. The framework coordinates an EV aggregator and an energy management system through endogenous DLMPs while incorporating load and photovoltaic (PV) uncertainty via a compact normal-minus-beta robust counterpart reformulation. An exactness lemma preserves the economic interpretation of DLMPs after KKT reformulation and Big-$M$ linearization. Furthermore, EV chargers provide reactive power support through non-unity-power-factor operation to improve voltage regulation. Numerical results on the IEEE 33-bus distribution system demonstrate improved voltage security and substantially lower computational complexity than conventional stochastic and robust optimization approaches. Future work will investigate battery degradation-aware scheduling, vehicle-to-grid operation, full AC optimal power flow formulations, unbalanced distribution networks, and distributed coordination algorithms.

\bibliographystyle{IEEEtran}
\bibliography{reference}

\end{document}